\documentclass{article}
\usepackage{longtable}
\usepackage{PRIMEarxiv}
\usepackage[utf8]{inputenc}
\usepackage[T1]{fontenc}
\usepackage{hyperref}
\usepackage{url}
\usepackage{booktabs}
\usepackage{amsfonts}
\usepackage{nicefrac}
\usepackage{microtype}
\usepackage{lipsum}
\usepackage{fancyhdr}
\usepackage{graphicx}
\graphicspath{{media/}}

\newcommand{\ubar}[1]{\text{\b{$#1$}}}

\usepackage{amsmath}
\usepackage{amsthm}
\newtheorem{theorem}{Theorem}
\newtheorem{corollary}[theorem]{Corollary}
\newtheorem{proposition}[theorem]{Proposition}

\theoremstyle{definition}
\newtheorem{definition}{Definition}
\newtheorem{assumption}{Assumption}

\usepackage{enumitem}
\usepackage{float}
\usepackage{algorithm}
\usepackage{algcompatible}

\usepackage[capitalise]{cleveref}
\crefname{assumption}{Assumption}{Assumptions}
\Crefname{assumption}{Assumption}{Assumptions}

\crefformat{enumi}{Assumption~#2#1#3}
\Crefformat{enumi}{Assumption~#2#1#3}

\title{Root Finding and Metamodeling for Rapid and Robust Computer Model Calibration}

\author{
  Yongseok Jeon,\quad Sara Shashaani \\
  Edward P. Fitts Department of Industrial and Systems Engineering \\
  North Carolina State University \\
  Raleigh, NC 27695 \\
  \texttt{\{yjeon, sshasha2\}@ncsu.edu}
}

\begin{document}

\maketitle

\begin{abstract}
We concern computer model calibration problem where the goal is to find the parameters that minimize the discrepancy between the multivariate real-world and computer model outputs. We propose to solve an approximation using signed residuals that enables a root finding approach and an accelerated search. We characterize the distance of the solutions to the approximation from the solutions of the original problem for the strongly-convex objective functions, showing that it depends on variability of the signed residuals across output dimensions, as wells as their variance and covariance. We develop a metamodel-based root finding framework under kriging and stochastic kriging that is augmented with a sequential search space reduction. We derive three new acquisition functions for finding roots of the approximate problem along with their derivatives usable by first-order solvers. Compared to kriging, the stochastic kriging metamodels will incorporate noisiness in the observations suggesting more robust solutions. We also analyze the case where a root may not exist. Our analysis of the asymptotic behavior in this context show that, since existence of roots in the approximation problem may not be known a priori, using new acquisition functions will not hurt the outcome. Numerical experiments on data-driven and physics-based examples demonstrate significant computational gains over standard calibration approaches.
\end{abstract}


\section{Introduction}\label{siam:sec:introduction}
With the recent advances in computational power and data acquisition technologies, computer models have become indispensable tools for simulation \cite{grieves2016digital}, prediction \cite{aivaliotis2019use}, and decision-support tools under uncertainty \cite{chaudhuri2023predictive}. However, the parameters that govern the behavior of a computer model often need to be carefully estimated to ensure accurate representation of the real-world system. This is commonly referred to as the \textit{model calibration problem}. Calibration involves finding the \emph{optimal} set of parameters that enable  the model outputs to match the observed outputs~\cite{gu2018scaled}. Significant progress has been made through various frameworks, including Bayesian calibration  \cite{kennedy2001bayesian, plumlee2017bayesian}, ordinary least squares \cite{tuo2015efficient}, maximum likelihood estimation \cite{sung2024review} or surrogate-based optimization \cite{zhan2022calibrating}.

Despite these advances, several practical challenges remain. When the observed dataset is large, calibration becomes computationally demanding. Moreover, due to the inherent stochasticity of the real-world systems, the calibrated parameters may overfit the observed data and perform poorly on unseen data. High-dimensional input and output spaces further complicate the problem by masking the influence of calibration parameters. Recent efforts to mitigate these challenges include subsampling strategies~\cite{lv2023fast}, parallel computing approaches~\cite{tang2024parallel}, and low-rank deep learning methods~\cite{bhatnagar2022computer}.

Building on this ongoing effort, the first objective of this paper is to accelerate the calibration process by means of approximations and reducing the search space. 
We propose to solve an approximation to the original discrepancy minimization objective function using a signed discrepancy metric. We show that finding the roots of this approximation, that may not be far from the true optimal calibration parameters under standard assumptions, can be more efficient. 
The second objective is to extend our proposed framework to noisy observations, where incorporating uncertainty promotes a more robust calibration.

\subsection{Problem Statement}
Let $\mathcal{D}_n = \{(x_j, y_j)\}_{j=1}^{n}$ be the dataset available up to time $n$, where $x_j \in \mathcal{X} \subset \mathbb{R}^{m_x}$ is the input variate that could  represent external factors such as environmental conditions and $y_j \in \mathcal{Y} \subset \mathbb{R}^{m_y}$ is the real-world system response. Throughout this work, we consider the computer model as a black-box function $h: \mathcal{X} \times \Theta \to \mathcal{Y}$ that maps the inputs to outputs with a parameter $\theta \in \Theta \subset \mathbb{R}^{m_\theta}$, where $\Theta$ is compact. Define the (random) residual function $R\in\mathbb{R}^{m_y}$ as $R(\theta) := Y - h(X;\theta)$ (removing $X$ and $Y$ from the input arguments for ease of exposition). 
Assuming that the system behavior up to time $n$ is in steady state, meaning that $\{(x_j, y_j)\}_{j=1}^{n}$ are independent and identically distributed samples from the joint distribution of $(X,Y)$, calibration seeks 
\begin{equation}\label{siam:eq:p} \tag{P}
    \min_{\theta \in \Theta}
    \bar{f}_{n}(\theta) 
    :=
    \frac{1}{m_y}\frac{1}{n}\sum_{j=1}^{n}\|r_j(\theta)\|_2^2
    ,
\end{equation}
where $r_j(\theta):= y_j - h(x_j;\theta)$ denotes the $j$-th residual of $\mathcal{D}_n$ and 
let $\bar{\Theta}_n^* := \arg\min_{\theta \in \Theta} \bar{f}_n(\theta)$ denote the set of optimal solutions.
The objective function in \eqref{siam:eq:p} is also referred to as empirical risk minimization---a sample average approximation of the problem,
\begin{equation}\label{siam:eq:sp} \tag{SP}
    \min_{\theta \in \Theta} \bar{f}(\theta)
    := 
    \frac{1}{m_y}\mathbb{E}\left[\|R(\theta)\|_2^2\right]
    ,
\end{equation} where the expectation is with respect to the unknown joint distribution of $(X,Y)$, 
and let $\bar{\Theta}^* := \arg\min_{\theta \in \Theta} \bar{f}(\theta)$ denote the set of optimal solutions.
Problem ~\eqref{siam:eq:sp} is a standard stochastic optimization problem that has commonly appeared for model calibration across many studies \cite{yuan2012calibration, xu2017model, liu2022parameter, cui2020calibrated}. There are several well-established solution methods for \eqref{siam:eq:sp} including stochastic gradient \cite{robbins1951stochastic}, stochastic trust-region \cite{shashaani2018astro}, and surrogate-based methods \cite{ankenman2010stochastic}. 

\subsection{Proposed Approximation and Root finding}
Consider an approximation of \eqref{siam:eq:sp} as
\begin{equation}\label{siam:eq:spaa}\tag{SAP}
    \min_{\theta \in \Theta} f(\theta)
    := 
    \frac{1}{m_y^2}
    \left(\mathbb{E}\left[
      \boldsymbol{1}^\top R(\theta)  
    \right]\right)^2
    = 
    \left(\mathbb{E}\left[
        \frac{1}{m_y}\sum_{i=1}^{m_y}[R(\theta)]_i
    \right]\right)^2,
\end{equation}
where the calibration objective is instead defined by aggregating the residuals across the output dimensions into a scalar, 
and let $\Theta^* := \arg\min_{\theta \in \Theta} f(\theta)$ denote the set of optimal solutions.
(Throughout the paper, for any vector $v\in\mathbb{R}^d$, we denote its $i$-th component by $[v]_i$.) This formulation satisfies the following relationship, 
\begin{equation*}\label{siam:eq:CS-Tr-gap}
    \mathbb{E}\left[\frac{1}{m_y}\|R(\theta)\|_2^2\right]
    \geq 
    \mathbb{E}\left[\left(\frac{1}{m_y}\sum_{i=1}^{m_y}[R(\theta)]_i\right)^2\right]
    \geq 
    \left(\mathbb{E}\left[\frac{1}{m_y}\sum_{i=1}^{m_y}[R(\theta)]_i\right]\right)^2,
\end{equation*}
where the first inequality follows from the Cauchy--Schwarz inequality, and the second inequality follows from the Jensen's inequality. Specifically, from the Cauchy--Schwarz inequality, 
$\left(\sum_{i=1}^{m_y}a_i b_i\right)^2 
    \leq
    \left(\sum_{i=1}^{m_y} a_i^2\right)
    \left(\sum_{i=1}^{m_y} b_i^2\right),$
and letting $a_i = 1$ and $b_i = [R(\theta)]_i$ for all $i = 1, \dots, m_y$, we obtain
$\left(\sum_{i=1}^{m_y} [R(\theta)]_i\right)^2 
    \leq
    m_y
    \left(\sum_{i=1}^{m_y} [R(\theta)]_i^2\right),$
which yields the first inequality after normalization and taking expectations. Therefore, solving the surrogate objective \eqref{siam:eq:spaa} amounts to minimizing the lower bound of the original problem \eqref{siam:eq:sp}. We propose to solve this approximation rather than the original problem which can be done more effectively with root finding for the following reason. To see this, define the averaged residual scalar
\begin{equation*}
    S(\theta) 
    := 
    \frac{1}{m_y}\boldsymbol{1}^\top R(\theta) 
    = 
    \frac{1}{m_y}\sum_{i=1}^{m_y}[R(\theta)]_i
    ,
\end{equation*}
which aggregates the residual vector $R(\theta)$ across the output dimensions into a single signed quantity. With this, $f(\theta) = (\mathbb{E}[S(\theta)])^2 \geq 0$, and we consider
\begin{equation*}\label{siam:eq:spaa-rf}\tag{RF}
    \tilde{f}(\theta)
    :=
    \mathbb{E}[S(\theta)]
    ,
    \quad
    \tilde{\Theta}^* 
    := 
    \left\{\theta \in \Theta \mid \tilde{f}(\theta)=0\right\}.
\end{equation*}
For a nonempty $\tilde{\Theta}^*$ (\Cref{siam:assump:rootexist-RF}), $f(\theta)$ is minimized at 0, and 
hence $\Theta^* = \tilde{\Theta}^*$,
thus solving \eqref{siam:eq:spaa} is equivalent to solving \eqref{siam:eq:spaa-rf}. Solving \eqref{siam:eq:spaa-rf} in the presence of the sign information allows us to bypass large portions of the parameter space and focus only on the regions where the sign is expected to change (see~\Cref{siam:fig:metamodel-min-vs-rf}).
This progressive search space reduction is illustrated in \Cref{siam:fig:rss,siam:fig:rss_stochastic}. Under the proposed framework, one can expect a more guided sampling trace during the search, as demonstrated in \Cref{siam:fig:2d_sampling_jrace}.

We primarily treat the root finding formulation of \eqref{siam:eq:spaa} as a surrogate to the original problem \eqref{siam:eq:sp}. 
Note that the approximate solution becomes tight in the special case of $m_y = 1$. In this case, $\mathbb{E}[(R(\theta))^2] = (\mathbb{E}[R(\theta)])^2 + \mathrm{Var}(R(\theta))$, so that the gap reduces to the variance term. In this setting, if either
\begin{enumerate}
    \item[(a)] the parameter that minimizes the problem \eqref{siam:eq:spaa} also minimizes the variance,  i.e., $\theta^* = \arg \min_{\theta \in \Theta}(\mathbb{E}[R(\theta)])^2 = \arg \min_{\theta \in \Theta}\mathrm{Var}(R(\theta)),$
    an outcome that arises when systems with optimal parameters stabilize and exhibit reduced variability~\cite{zhu2018simulation, jacobson2021exploring}, or
    
    \item[(b)] the variance satisfies the relation $\mathrm{Var}\left(R(\theta)\right)\propto \frac{d}{d\theta}\mathbb{E}\left[R(\theta)\right]$, which is known as state-dependent variance property where the noise diminishes near stationary points,
\end{enumerate}
then solving \eqref{siam:eq:spaa} also solves \eqref{siam:eq:sp}, i.e., 
$\arg\min_{\theta\in \Theta} \left(\mathbb{E}\left[R(\theta)\right]\right)^2
    =
    \arg\min_{\theta\in \Theta} \mathbb{E}[(R(\theta))^2]
    .$
Beyond this special case, solving the approximation problem, in general, leads to a different solution. This naturally leads to the question: how accurate is the resulting approximate solution? We address this question in \Cref{siam:sec:sol-acc-rf}; in summary, we show that the solution bound is characterized by the residual properties, and discuss the stochastic treatment to further refine this accuracy.

\subsection{Metamodeling and Root finding} 
When solving \eqref{siam:eq:spaa-rf}, 
we employ a metamodeling approach (e.g., kriging or GPR), and solve the problem via \textit{Bayesian Optimization}, but with adjusted acquisition functions for the root finding scheme. One could alternatively use other approaches for root finding with a stochastic derivative-free optimization (DFO) algorithm such as \cite{shashaani2018astro, larson2019derivative}. An advantage of using kriging is the ability to extend, with some modifications, to  stochastic kriging~\cite{staum2009better} in order to take uncertainty in the observed residuals into account. In the stochastic setting, the metamodel's output remains random even at the evaluated parameter values, allowing us to seek solutions that are optimal in expectation rather than only for a fixed data sample. This stochastic treatment also provides an opportunity to further tighten the solution accuracy bound analyzed in \Cref{siam:sec:sol-acc-rf}. 
We demonstrate through several examples that stochastic kriging finds better solutions
to the original problem than deterministic kriging under the same computational budget.

Additionally, the sample average approximation of \eqref{siam:eq:spaa-rf} may not have a root--where all aggregated residuals are either positive or negative over the parameter space--due to the finite sample variability or misspecification of the computer model.
We show in \Cref{siam:sec:rootless} that the proposed root finding framework naturally retrieves the optimizer of this sample average even in the absence of a root.

\subsection{Contributions and Paper Organization} Our contributions in this paper are threefold. 
First, we analyze the solution accuracy of the root finding approximation and characterize its error bounds and behavior compared to the original problem. We also show that a stochastic approach that accounts for the finite-sample error can further improve the solution accuracy. Second, we develop a metamodeling framework tailored to the root finding, introducing new acquisition functions and a search space reduction technique, and characterize its asymptotic behavior in the absence of a root. Lastly, we extend the framework to stochastic setting to enhance robustness against noisy observations. We discuss that stochastic kriging can be used as the stochastic approach to further improve the solution accuracy.

The paper is organized as follows. \Cref{siam:sec:sol-acc-rf} analyzes the solution accuracy of the root finding approximation and derives error bounds with convergence analysis. \Cref{siam:sec:det-metamodeling} reviews kriging-based metamodeling and standard acquisition strategies for deterministic calibration. \Cref{siam:sec:rf-metamodeling} introduces the proposed root finding framework with modified acquisition strategies and acceleration techniques, and characterizes the asymptotic behavior of these strategies when no root exists.
\Cref{siam:sec:sto-metamodeling} extends the root finding framework to stochastic kriging for robust calibration. \Cref{siam:sec:experiment} presents the numerical experiments and concludes in \Cref{siam:sec:conclusion} with several remarks. 

\subsection*{List of Notations}
\begin{longtable}{r@{\;$\hat{=}$\;}l}
    $m_x, m_y, m_\theta$ & dimensions of input, output, and parameter space \\
    $X,Y$ & random input and output variates\\
    $\mathcal{D}_n = \{(x_j, y_j)\}_{j=1}^n$ & observed dataset of size $n$ \\
    $h:\mathcal{X} \times \Theta \to \mathcal{Y}$ & computer model output \\
    $\theta \in \Theta \subset \mathbb{R}^{m_\theta}$ & calibration parameter \\
    $R(\theta) = Y - h(X;\theta) \in \mathbb{R}^{m_y}$ & residual vector\\
    $r_j(\theta) = y_j - h(x_j;\theta)$ & $j$-th residual of $\mathcal{D}_n$ \\
    $S(\theta) = \frac{1}{m_y}\boldsymbol{1}^\top R(\theta)$ & averaged residual scalar \\
    $\bar{f}(\theta), f(\theta), \tilde{f}(\theta)$ & objectives of \eqref{siam:eq:sp}, \eqref{siam:eq:spaa}, \eqref{siam:eq:spaa-rf} \\
    $\bar{f}_n(\theta)$, $f_n(\theta)$, $\tilde{f}_n(\theta)$ & sample average approximations of \eqref{siam:eq:sp}, \eqref{siam:eq:spaa} and \eqref{siam:eq:spaa-rf}\\
    $\delta(\theta)$ & spatial variability; see \eqref{siam:def:spatial-variability}\\
    $\bar{\Theta}^*, \Theta^*$ & solution sets of \eqref{siam:eq:sp} and \eqref{siam:eq:spaa}\\
    $\Theta_n^*, \check{\Theta}_n^*$ & deterministic and stochastic approximate solution sets for $\Theta^*$\\
    $\Theta^t = \{\theta_1,\dots,\theta_{p+t}\}$ & set of configurations at iteration $t$ with $p$ initial design points
\end{longtable}

\section{Solution Accuracy of the Surrogate Approximation}\label{siam:sec:sol-acc-rf}
In this section, we analyze the solution accuracy of the surrogate approximation \eqref{siam:eq:spaa} relative to the original problem \eqref{siam:eq:sp}. We characterize the gap between the solutions through the structural properties of \eqref{siam:eq:sp}, the spatial variability of the residuals, the variance of the aggregated residuals, and an estimation error for solving \eqref{siam:eq:spaa} with a finite sample $\mathcal{D}_n$. This gap becomes small when the residuals are approximately uniform across output dimensions and have low variability across observations, and when the function \eqref{siam:eq:sp} increases more rapidly away from its minimizer $\bar{\theta}^*$. Additionally, this gap can be further tightened by accounting for the finite sample uncertainty.

We begin by characterizing the approximation error arising from replacing $\bar{f}(\theta)$ with its 
minorizing function $f(\theta)$. From \eqref{siam:eq:sp} and \eqref{siam:eq:spaa}, we can write 
\begin{align*}
    \bar{f}(\theta) - f(\theta)
    &= 
    \mathbb{E}\left[\frac{1}{m_y}\|R(\theta)\|_2^2\right]
    - 
    \left(\mathbb{E}\left[S(\theta)\right]\right)^2
    \\
    &=
    \underbrace{\mathbb{E}\left[\frac{1}{m_y}\|R(\theta)\|_2^2\right] - \mathbb{E}\left[\left(S(\theta)\right)^2\right]}_{:= \delta(\theta)}
    +    \underbrace{\mathbb{E}\left[\left(S(\theta)\right)^2\right] - \left(\mathbb{E}\left[S(\theta)\right]\right)^2}_{=\mathrm{Var}(S(\theta))}.
\end{align*}
We refer to the first term as the spatial variability and define it as follows.
\begin{definition}[Spatial Variability]\label{siam:def:spatial-variability}
    The spatial variability $\delta(\theta)$ is defined as
    \begin{align*}
        \delta(\theta)
        &:=
        \mathbb{E}\left[
            \frac{1}{m_y}\sum_{i=1}^{m_y}
                ([R(\theta)]_i)^2
                \right]
        -
        \mathbb{E}\left[
                \left(\frac{1}{m_y}\sum_{i=1}^{m_y}[R(\theta)]_i\right)^2
        \right]
        \\
        &= 
        \mathbb{E}\left[
            \frac{1}{m_y}\sum_{i=1}^{m_y}\left(
                [R(\theta)]_i - \frac{1}{m_y}\sum_{j=1}^{m_y}[R(\theta)]_j
            \right)^2
        \right]
    \end{align*}
\end{definition} 
$\delta(\theta)$ quantifies the variability of the residuals across the output dimensions. This term vanishes when $m_y = 1$, and becomes small when, within each observation, the residual components are similar. This can occur in systems where the model outputs are governed by a common underlying mechanism and operate under similar conditions.

The second term, $\mathrm{Var}(S(\theta))$, represents the variability of the aggregated residual across observation. This term can be expressed as
\begin{equation*}
    \mathrm{Var}(S(\theta)) 
    = 
    \frac{1}{m_y^2}\left(\sum_{i=1}^{m_y}\mathrm{Var}([R(\theta)]_i) 
    + \sum_{i\neq j}\mathrm{Cov}([R(\theta)]_i, [R(\theta)]_j)\right)
    ,
\end{equation*}
and is small when the residuals in each output dimension have low variability across observations, or when the residual components are negatively correlated. The conditions that make small $\delta(\theta)$ and $\mathrm{Var}(S(\theta))$ are conceptually different, and thereby neither condition implies the other. However, both terms can be small when the residuals are similar across output dimensions and exhibit low variability across the observations. Notably, large values of these terms may limit the accuracy of the proposed approximation, but this also reflects the inherent difficulty of the original problem \eqref{siam:eq:sp} itself. 

We next discuss the structural properties of \eqref{siam:eq:sp} that are needed to characterize the solution accuracy bound. Since $\bar f$ is a composite function that depends on $h$, $\theta$, and $(X,Y)$, we impose the following regularity conditions on these components.
\begin{assumption}[Regularity Conditions on the Computer Model]\label{siam:assump:reg:h}~
    \begin{enumerate}[label=(\alph*), ref=\theassumption(\alph*)]
        \item\label{siam:assump:C2} $h(X;\theta)$ is twice-continuously differentiable with respect to $\theta$ for all $\theta \in \Theta$. Accordingly, let  $J_h(\theta) := \nabla_{\theta}h(X;\theta) \in \mathbb{R}^{m_y\times m_\theta}$ and $H_h(\theta) := \nabla_{\theta}^2 h(X;\theta) \in \mathbb{R}^{m_y \times m_\theta \times m_\theta}$ denote the Jacobian and Hessian of $h$ with respect to $\theta$.

        \item \label{siam:assump:V}
        There exists a non-negative integrable random variable $V$ with $\mathbb{E}[V] < \infty$ such that $\|J_h(\theta)^\top R(\theta)\|_2 \leq V$ for all $\theta \in \Theta$.

        \item\label{siam:assump:W} There exists a non-negative integrable random variable $W$ with $\mathbb{E}[W] < \infty$ such that $\|R(\theta)\|_2^2 \leq W$ for all $\theta \in \Theta$.
        
        \item\label{siam:assump:U} There exists a non-negative integrable random variable $U$ with $\mathbb{E}[U] < \infty$ such that $\left\|J_h(\theta)^\top J_h(\theta) - \sum_{i=1}^{m_y}[R(\theta)]_i[H_h(\theta)]_i\right\|_F \leq U$ for all $\theta \in \Theta$, where $\|\cdot\|_F$ denotes the Frobenius norm.

        \item \label{siam:assump:eigen}There exists $\bar{u} > \mathbb{E}[U]$, where $U$ is as in \Cref{siam:assump:U} and $\lambda_{\min}$ denotes the smallest eigenvalue,  such that $\inf_{\theta \in \Theta} \lambda_{\min}\left(
            \mathbb{E}[J_h(\theta)^\top J_h(\theta)]
        \right) \geq \bar{u}$.
    \end{enumerate}
\end{assumption}
As $J_h(\theta)$ and $H_h(\theta)$ are random due to their dependence on $X$, we invoke the following result to differentiate $\bar{f}(\theta)$ under the expectation.
\begin{theorem}[Lebesgue's Dominated Convergence Theorem]\label{siam:thm:DCT}
    Assume $X_n \to X$ almost surely as $n \to \infty$. Suppose there exists $Y \in L^1(\mathbb{P})$ such that
    $|X_n|\leq Y$  almost surely for every $n$.
    Then $X \in L^1(\mathbb{P})$ and
    \[
        \mathbb{E}[X_n] \to \mathbb{E}[X] \quad \text{as } n \to \infty.
    \]
\end{theorem}
Under \Cref{siam:assump:W}, $R(\theta) \in L^2$ and under \Cref{siam:assump:V} and \Cref{siam:thm:DCT}, the gradient of $\bar{f}(\theta)$ is
\begin{equation*}
    \nabla_{\theta}\bar{f}(\theta) 
    = \frac{1}{m_y}\nabla_{\theta} \mathbb{E}\left[\|R(\theta)\|_2^2\right] 
    =  \frac{1}{m_y}\mathbb{E}\left[\nabla_{\theta}\|R(\theta)\|_2^2\right] 
    = -\frac{2}{m_y}\mathbb{E}[J_h(\theta)^\top R(\theta)]\in \mathbb{R}^{m_\theta}.
\end{equation*}
Similarly, under \Cref{siam:assump:U} and \Cref{siam:thm:DCT}, the Hessian of $\bar{f}(\theta)$ is
\begin{align*}
    \nabla_{\theta}^2 \bar f(\theta)
    &=
    \frac{1}{m_y}\nabla_{\theta}^2 \mathbb{E}\left[\|R(\theta)\|_2^2\right]
    =
    -\frac{2}{m_y}\nabla_{\theta}\mathbb{E}\left[J_h(\theta)^\top R(\theta)\right]
    \\
    &=
    \frac{2}{m_y}\mathbb{E}\left[
    J_h(\theta)^\top J_h(\theta)
    -\sum_{i=1}^{m_y}[R(\theta)]_i[H_h(\theta)]_i
    \right] \in \mathbb{R}^{m_\theta\times m_\theta}
    .
    \nonumber
\end{align*}
Then, under \Cref{siam:assump:eigen}, for any $\theta \in \Theta$,
\begin{align*}
    \nabla_{\theta}^2 \bar{f}(\theta) 
    \succeq 
    \frac{2}{m_y}\left(
        \lambda_{\min}\left(\mathbb{E}[J_h(\theta)^\top J_h(\theta)]\right) 
        - \mathbb{E}[U]
    \right)I_{m_\theta} 
    \succeq 
    \frac{2}{m_y}(\bar{u} - \mathbb{E}[U])I_{m_\theta} 
    \succ 0
    ,
\end{align*}
where we obtain the following strong convexity condition for \eqref{siam:eq:sp}.
\begin{corollary}[Strong Convexity of the Calibration Objective]\label{siam:cor:strong-convex}
    Under \Cref{siam:assump:reg:h}, $\bar{f}(\theta)$ is $\bar{\ell}_C$-strongly convex over $\Theta$ with $\bar{\ell}_C = \frac{2}{m_y}(\bar{u} - \mathbb{E}[U])$, that is,
    \[
        \bar{f}(\theta) - \bar{f}(\bar{\theta}  ^*) \geq \frac{\bar{\ell}_C}{2} \|\theta - \bar{\theta}^* \|^2,
    \]
    for all $\theta \in \Theta$, and $\bar{\theta}^*$ is the unique minimizer of $\bar{f}(\theta)$, and hence $\bar{\Theta}^* = \{\bar{\theta}^*\}$.
\end{corollary}
We then define the distance between a point $a \in \mathbb{R}^d$ and a set $\mathcal{B} \subset \mathbb{R}^d$ as $\mathrm{dist}(a,\mathcal{B}) := \min_{b \in \mathcal{B}}\|a-b\|$, so that under \Cref{siam:cor:strong-convex}, solving \eqref{siam:eq:spaa} gives
\begin{equation*}
    \mathrm{dist}(\bar{\theta}^*, \Theta^*)^2
    \leq
    \frac{2}{\bar{\ell}_C}
    \min_{\theta \in \Theta^*}
    \left(
        \bar{f}(\theta) - \bar{f}(\bar{\theta}^*)
    \right).
\end{equation*}
Since $\Theta^*$ is not directly accessible and can only be estimated from the finite sample $\mathcal{D}_n$, let the sample average approximation of \eqref{siam:eq:spaa} be
\begin{equation}\label{siam:eq:spaa-saa}
    f_n(\theta)
    :=
    \Bigg(\frac{1}{m_y}\frac{1}{n}\sum_{j=1}^n \boldsymbol{1}^\top r_j(\theta)\Bigg)^2
    ,
    \quad
    \Theta_n^* 
    := 
    \arg\min_{\theta \in \Theta} f_n(\theta)
    ,
\end{equation}
where $\Theta_n^*$ denotes the set of optimal solutions. To study the behavior of $\Theta_n^*$, we next establish the uniform convergence of $f_n$ to $f$ over $\Theta$.
\begin{proposition}[Uniform Law of Large Numbers]\label{siam:prop:ulln}
    Given that $\mathcal{D}_n$ consists of i.i.d. draws from the joint distribution $(X,Y)$, $\Theta$ is compact and under \Cref{siam:assump:W},
    \[
        \sup_{\theta \in \Theta} 
        \left|f_n(\theta) - f(\theta)\right| 
        \to 0 
        \quad \text{as } n \to \infty \quad \text{a.s.}
    \]
\end{proposition}
The integrability condition needed for \Cref{siam:prop:ulln}~\cite[Pro.~8.5]{kim2014guide} is satisfied under \Cref{siam:assump:W}. That is, from Cauchy--Schwarz inequality, for all $\theta \in \Theta$, 
\begin{equation*}
    \left(
        \frac{1}{m_y}\boldsymbol{1}^\top r_j(\theta)
    \right)^2
    \leq
    \frac{1}{m_y}\|r_j(\theta)\|_2^2
    ,
    \quad
    j = 1, \dots, n.
\end{equation*}
Since $r_j(\theta)$ is a realization of $R(\theta)$, \Cref{siam:assump:W} gives for all $\theta \in \Theta$,
\begin{equation*}
    \frac{1}{m_y}\|R(\theta)\|_2^2
    \leq
    \frac{1}{m_y}W
    .
\end{equation*}
Hence, the sample average terms are bounded by the integrable random variable $W/m_y$, and since $\mathbb{E}[W] < \infty$, the integrability condition holds. Note that \Cref{siam:prop:ulln} itself does not imply that solving \eqref{siam:eq:spaa-saa} yields a solution that is close to $\Theta^*$. That said, even with $f_n(\theta)$ uniformly close to $f(\theta)$, a minimizer of $f_n$ may remain far from $\Theta^*$ if $f$ is flat near $\Theta^*$. Thus, convergence of a minimizer of $f_n$ to the solution set $\Theta^*$ further requires a growth condition near $\Theta^*$. For this, we next impose the following regularity conditions.
\begin{assumption}[Regularity Conditions on the Root-Finding Approximation]\label{siam:assump:reg:rf}~
    \begin{enumerate}[label=(\alph*), ref=\theassumption(\alph*)]
    \item\label{siam:assump:rootexist-RF} $\tilde{\Theta}^*$ is nonempty.

    \item\label{siam:assump:root-stab}
    There exists $\tilde{\ell}_R > 0$ such that for every $\theta_a \in \Theta$ and every $\theta_b \in \arg\min_{\theta \in \tilde{\Theta}^*}\|\theta_a - \theta\|$, $|\tilde f(\theta_a)-\tilde f(\theta_b)| \geq \tilde{\ell}_R \|\theta_a-\theta_b\|$.
    
    \item\label{siam:assump:rate-grad2}
    $\mathbb{E}\left[\left\|\nabla_\theta (S(\theta))^2\right\|_2^2\right]< \infty$ for all $\theta \in \Theta.$

    \item\label{siam:assump:rate-gradlip}
    There exists a nonnegative random variable $L$ with $\mathbb{E}[L^2] < \infty$ such that $\|\nabla_\theta (S(\theta_a))^2-\nabla_\theta (S(\theta_b))^2\|_2 \leq L \|\theta_a-\theta_b\|_2$ for all $\theta_a, \theta_b \in \Theta$.

    \item\label{siam:assump:lips} $\tilde{f}(\theta)$ is $\tilde{\ell}_E$-Lipschitz continuous over $\Theta$, i.e., $|\tilde{f}(\theta_a) - \tilde{f}(\theta_b)| \leq \tilde{\ell}_E\|\theta_a - \theta_b\|$ for all $\theta_a, \theta_b \in \Theta$.

    \item\label{siam:assump:var-lips} $\mathrm{Var}(S(\theta))$ is $\tilde{\ell}_V$-Lipschitz continuous over $\Theta$, i.e., $\left|\mathrm{Var}(S(\theta_a)) - \mathrm{Var}(S(\theta_b))\right| \leq \tilde{\ell}_V\|\theta_a - \theta_b\|$ for all $\theta_a, \theta_b \in \Theta$. 
    
    \item\label{siam:assump:out-var-lip} $\delta(\theta)$ is $\tilde{\ell}_\delta$-Lipschitz continuous over $\Theta$, i.e., $\left|\delta(\theta_a) - \delta(\theta_b)\right| \leq \tilde{\ell}_\delta\|\theta_a - \theta_b\|$ for all $\theta_a, \theta_b \in \Theta$. 
    \end{enumerate}
\end{assumption}
We next establish the convergence rate of sample average approximation.
\begin{theorem}[Convergence Rate of Sample Average Approximation]\label{siam:thm:saa-rate}
    Given that $\mathcal{D}_n$ consists of i.i.d. draws from the joint distribution $(X,Y)$ and $\Theta$ is compact, under \Cref{siam:assump:W}, \Cref{siam:assump:rootexist-RF}, \Cref{siam:assump:root-stab}, \Cref{siam:assump:rate-grad2}, and \Cref{siam:assump:rate-gradlip}, for any $\theta_n^* \in \Theta_n^*$,
    \[
        \mathrm{dist}(\theta_n^*,\Theta^*) = O_p(n^{-1/2}),
    \]
    where $O_p(n^{-1/2})$ denotes that $n^{1/2}\mathrm{dist}(\theta_n^*,\Theta^*)$ is bounded in probability.
\end{theorem}
\Cref{siam:thm:saa-rate} requires a quadratic growth condition for $f$ near the solution set $\Theta^*$; see \cite[Thm.~8.4]{kim2014guide} and \cite[Thm.~5.21]{van2000asymptotic}. Under \Cref{siam:assump:rootexist-RF}, we have $\Theta^*=\tilde{\Theta}^*$ and $\tilde f(\theta)=0$ for all $\theta \in \tilde{\Theta}^*$. Fix any $\theta_a \in \Theta$, and let $\theta_b \in \arg\min_{\theta \in \tilde{\Theta}^*}\|\theta_a-\theta\|$. Then, by \Cref{siam:assump:root-stab},
\begin{equation*}
    |\tilde f(\theta_a)-\tilde f(\theta_b)|
    \geq
    \tilde{\ell}_R \|\theta_a-\theta_b\|
    =
    \tilde{\ell}_R \mathrm{dist}(\theta_a,\tilde{\Theta}^*).
\end{equation*}
Since $\tilde f(\theta_b)=0$, it follows that $|\tilde f(\theta_a)| \geq \tilde{\ell}_R \mathrm{dist}(\theta_a,\tilde{\Theta}^*)$ and therefore,
\begin{equation*}
    f(\theta_a)
    =
    (\tilde f(\theta_a))^2
    \geq
    \tilde{\ell}_R^2 \mathrm{dist}(\theta_a,\Theta^*)^2,
\end{equation*}
so that the quadratic growth condition for $f$ near $\Theta^*$ is satisfied. Under \Cref{siam:assump:C2}, $R(\theta)$ is twice continuously differentiable in $\theta$, and hence so is $S(\theta)=\frac{1}{m_y}\mathbf{1}^\top R(\theta)$. Therefore, $\nabla_\theta (S(\theta))^2 \in \mathbb{R}^{m_\theta}$ is well-defined for all $\theta \in \Theta$. \Cref{siam:assump:rate-grad2} and \Cref{siam:assump:rate-gradlip} further impose regularity conditions on this gradient. We additionally impose the following assumption to derive convergence rates in expectation. 
\begin{assumption}[Bounded Second Moment of the Estimation Error]\label{siam:assump:BS2}
    For any $\theta^*_n \in \Theta^*_n$, $\sup_{n \geq 1}\mathbb{E}[n~ \mathrm{dist}(\theta^*_n, \Theta^*)^2] < \infty$.
\end{assumption}

\begin{corollary}[Sample Average Convergence Rates in Expectation]
Under the assumptions of \Cref{siam:thm:saa-rate} and \Cref{siam:assump:BS2},
    \[
        \mathbb{E}[\mathrm{dist}(\theta^*_n, \Theta^*)] = O(n^{-1/2})
        ,
        \quad
        \mathbb{E}[\mathrm{dist}(\theta^*_n, \Theta^*)^2] = O(n^{-1}).
    \]
\end{corollary}

\begin{proof}
    By \Cref{siam:assump:BS2}, there exists a constant $c>0$ such that for all $n \geq 1$,
    \begin{equation*}
        \mathbb{E}\left[n\mathrm{dist}(\theta_n^*,\Theta^*)^2\right] 
        \leq 
        c.
    \end{equation*}
    Hence, $\mathbb{E}\left[\mathrm{dist}(\theta^*_n,\Theta^*)^2\right] \leq \frac{c}{n}$, which gives
    \begin{equation*}
        \mathbb{E}\left[\mathrm{dist}(\theta^*_n,\Theta^*)^2\right] 
        = 
        O(n^{-1}).
    \end{equation*}
    Also, by the Cauchy--Schwarz inequality,
    \begin{equation*}
        \mathbb{E}\left[\mathrm{dist}(\theta^*_n,\Theta^*)\right]
        \leq
        \sqrt{\mathbb{E}\left[\mathrm{dist}(\theta^*_n,\Theta^*)^2\right]}
        = O(n^{-1/2}).
    \end{equation*}
\end{proof}
We then establish the solution accuracy bound for the sample average approximation.
\begin{theorem}[Solution Accuracy Bound]\label{siam:thm:sol-acc}
    Given that $\mathcal{D}_n$ consists of i.i.d. draws from the joint distribution of $(X,Y)$, $\Theta$ is compact, and under \Cref{siam:assump:reg:h} and \Cref{siam:assump:reg:rf}
    \begin{align*}
        \mathrm{dist}(\bar{\theta}^*, \Theta_n^*)^2
        &\leq
        \frac{2}{\bar{\ell}_C}
        \Bigg(
            \tilde{\ell}_E^2\mathrm{dist}(\theta_n^*, \Theta^*)^2
            +
            \sup_{\theta \in \Theta^*}\delta(\theta)
            +
            \tilde{\ell}_\delta\mathrm{dist}(\theta_n^*, \Theta^*)
        \\
        &\quad
            + \sup_{\theta \in \Theta^*}\mathrm{Var}(S(\theta))
            + \tilde{\ell}_V\mathbb{E}\left[\mathrm{dist}(\theta_n^*, \Theta^*)\right]
            + \tilde{\ell}_E^2\mathbb{E}\left[\mathrm{dist}(\theta_n^*, \Theta^*)^2\right]
        \Bigg).
        \nonumber
\end{align*}
\end{theorem}
For the proof, see \Cref{siam:appendix:sol-acc-proof}. We then establish the convergence rate of this bound. 
\begin{corollary}[Convergence Rate of Solution Accuracy]\label{siam:cor:conv-rate}
    Under the assumptions of \Cref{siam:thm:sol-acc} and \Cref{siam:assump:BS2},
    \begin{equation*}
        \mathrm{dist}(\bar{\theta}^*, \Theta_n^*)^2
        \leq
        \frac{2}{\bar{\ell}_C}
        \sup_{\theta \in \Theta^*}
        \left(\delta(\theta) + \mathrm{Var}(S(\theta))\right)
        + \mathcal{O}_p(n^{-1/2}).
    \end{equation*}
\end{corollary}
\begin{proof}
    Follows directly from \Cref{siam:thm:sol-acc} by applying \Cref{siam:thm:saa-rate} and \Cref{siam:assump:BS2}.
\end{proof}
\Cref{siam:cor:conv-rate} shows that the accuracy of the surrogate solution is governed by the spatial variability, the variance of the aggregated residual, and the curvature of the original objective. In particular, solving \eqref{siam:eq:spaa} yields a more reliable proxy for \eqref{siam:eq:sp} when the variability terms are small, or when $\bar{f}$ is sufficiently strongly curved around its minimizer so that the effect of these variability terms is small. Importantly, these components are intrinsic to the calibration problem itself and cannot be reduced through \eqref{siam:eq:spaa}. However, a better treatment of finite-sample approximation can be obtained by accounting for sampling uncertainty, where $\mathcal{D}_n$ is treated as random. To formalize this, we impose the following assumption.
\begin{assumption}[Improved Uniform Mean Squared Error of the Stochastic Estimator]\label{siam:assump:better-mse}
Define $\check{f}_n$ as a stochastic estimator of $f$ constructed from $\mathcal{D}_n$, and $\check{\Theta}_n^* := \arg\min_{\theta \in \Theta} \check f_n(\theta).$ Then,
\[
    \mathbb{E}\left[\sup_{\theta \in \Theta}(\check f_n(\theta)-f(\theta))^2\right]
    \leq
    \mathbb{E}\left[\sup_{\theta \in \Theta} (f_n(\theta)-f(\theta))^2\right].
\]
\end{assumption}
\Cref{siam:assump:better-mse} states that the stochastic estimator $\check{f}_n$ provides better estimates of $f$ than $f_n$ over $\Theta$. In \Cref{siam:sec:sto-metamodeling}, we show that stochastic kriging motivates \Cref{siam:assump:better-mse} under suitable regularity conditions. The next result gives an expected version of the solution accuracy bound.
\begin{corollary}[Expected Solution Accuracy Bound]\label{siam:cor:exp-sol-acc}
    Under the assumptions of \Cref{siam:thm:sol-acc} and \Cref{siam:assump:root-stab},
    \begin{align*}
        \mathbb{E}\left[\mathrm{dist}(\bar{\theta}^*, \Theta_n^*)^2\right]
        &\leq
        \frac{2}{\bar{\ell}_C}
        \Bigg(
            2\tilde{\ell}_E^2\mathbb{E}\left[\mathrm{dist}(\theta_n^*, \Theta^*)^2\right]
            +
            \sup_{\theta \in \Theta^*}\delta(\theta)
            +
            \sup_{\theta \in \Theta^*}\mathrm{Var}(S(\theta))
        \\
        &\quad
            + (\tilde{\ell}_\delta + \tilde{\ell}_V)
            \left(\mathbb{E}\left[\mathrm{dist}(\theta_n^*, \Theta^*)^2\right]\right)^{1/2}
        \Bigg).
    \end{align*}
\end{corollary}
See \Cref{siam:appendix:exp-sol-acc-proof} for the proof. The following result shows that an improvement in the mean squared estimation error of the stochastic estimator can lead to a tighter upper bound on the expected solution accuracy.
\begin{corollary}[Improved Upper Bound on Expected Solution Accuracy]\label{siam:cor:tight-bound}
Let $\overline{\mathbb{E}}[\mathrm{dist}(\bar{\theta}^*, \mathcal{A})^2]$ denote the upper bound on $\mathbb{E}[\mathrm{dist}(\bar{\theta}^*, \mathcal{A})^2]$ given in \Cref{siam:cor:exp-sol-acc}. Under the assumptions of \Cref{siam:cor:exp-sol-acc} and \Cref{siam:assump:better-mse},
\[
    \overline{\mathbb{E}}\left[
        \mathrm{dist}(\bar{\theta}^*, \check{\Theta}_n^*)^2
    \right]
    \leq
    \overline{\mathbb{E}}\left[
        \mathrm{dist}(\bar{\theta}^*, \Theta_n^*)^2
    \right].
\]
\end{corollary}
For the proof, see \Cref{siam:appendix:tight-bound-proof}. We note that \Cref{siam:cor:tight-bound} does not necessarily imply a comparable convergence rate as shown in \Cref{siam:thm:sol-acc}. Establishing such a rate would require additional assumptions on $\check f_n$, analogous to those used in \Cref{siam:thm:saa-rate}. We do not pursue that extension here, although one may expect a similar type of rate result under suitable regularity conditions. We note that several assumptions used in this section, such as \Cref{siam:assump:root-stab}, are imposed globally over $\Theta$ rather than locally near the optimum, primarily for technical convenience. The key implication of this section is that \eqref{siam:eq:spaa} can become a reliable proxy for \eqref{siam:eq:sp}, and that a stochastic treatment can further improve the solution accuracy bound.

\section{Calibration with Deterministic Simulation Metamodeling}\label{siam:sec:det-metamodeling}
Solving the calibration problem \eqref{siam:eq:p} can be computationally challenging, especially when $\mathcal{D}_n$ is large or when the evaluation of $h(\cdot)$ itself is time consuming. This motivates the use of a simulation metamodeling approach, often referred to as the \textit{response surface} method, in which a surrogate model is constructed to approximate the objective function.

\subsection{Kriging}
Kriging is a widely-used metamodeling approach that was originally developed in geostatistics for spatial interpolation, such as estimating mineral distributions in mining \cite{oliver1990kriging}. It has since then evolved into a standard technique with applications in machine learning, uncertainty quantification, and simulation metamodeling. In the machine learning area, kriging is more commonly referred to as Gaussian process regression (GPR) and is widely used for hyperparameter tuning, prediction and uncertainty quantification. To ease the use of Kriging for solving the calibration task in \eqref{siam:eq:p}, we generally assume that $\bar{f}_{n}(\theta)$ is a realization of a Gaussian random field (GRF) over $\Theta$ \cite{staum2009better}.

Define $\Theta^t$ as the set of evaluated parameters at iteration $t$. To initialize the procedure, we allow for an initial design consisting of $p$-points in the parameter space, so that $\Theta^t$ consists of these $p$-initial design points when $t=0$, i.e., $\Theta^0 = \{\theta_1, \dots, \theta_p\}$. Let $\bar{\boldsymbol{f}}^t_{n}$ represents the corresponding function values at $\Theta^t$. For solving \eqref{siam:eq:p}, we get $\bar{\boldsymbol{f}}^t_{n} = [\bar{f}_{n}(\theta_1), \dots, \bar{f}_{n}(\theta_{p+t})]^\top$. Then, the predictive distribution of kriging at an unseen configuration $\theta$ at iteration-$t$ is expressed as:
\begin{align}\label{siam:eq:kriging-pred-dist}
    \eta^t_{n}(\theta) &\sim \mathcal{N}\left(\mu^t_{n}(\theta), (\sigma^t_{n})^2(\theta)\right),\\
    \mu^t_{n}(\theta) &= \mathbf{k}(\theta, \Theta^t;l) [\mathbf{k}(\Theta^t,\Theta^t;l)]^{-1} \bar{\boldsymbol{f}}^t_{n}, \nonumber\\
    (\sigma^t_{n})^2(\theta) &= k(\theta, \theta;l) - \mathbf{k}(\theta,\Theta^t;l) [\mathbf{k}(\Theta^t,\Theta^t;l)]^{-1} \mathbf{k}(\theta,\Theta^t;l)^\top \nonumber,
\end{align}
where $\Theta^t = \{\theta_1, \dots, \theta_{p+t}\}$. The kernel function $k(\cdot, \cdot; l)$ defines the correlation between two configurations in the parameter space.  A common choice for kernel function is squared exponential kernel, also known as radial basis function (RBF) kernel, which is defined by $k(\theta,\theta';l) = \exp\left(-\frac{\|\theta-\theta'\|^2}{2l^2}\right),$
where $\theta, \theta' \in \mathbb{R}^{m_\theta}$, and $l > 0$ is the length-scale hyperparameter for the kernel. Given two sets of points $w \in \mathbb{R}^{m_m \times m_\theta}$ and $v \in \mathbb{R}^{m_n \times m_\theta}$, we define the kernel matrix $\mathbf{k}(w,v;l) \in \mathbb{R}^{m_m \times m_n}$, where each entry corresponds to the kernel evaluation between a pair of points from the two sets. The $(i,j)-$th entry of the kernel matrix $\mathbf{k}(w,v;l)$ is $[\mathbf{k}(w,v;l)]_{i,j} := k([w]_i, [v]_j;l).$
$\mathbf{k}(\theta, \Theta^t ;l) \in \mathbb{R}^{p+t}$ denotes the vector of kernel evaluations between the test point $\theta$ and the set of evaluated points $\Theta^t \in \mathbb{R}^{(p+t) \times m_\theta}$. The kernel matrix $\mathbf{k}(\Theta^t, \Theta^t;l) \in \mathbb{R}^{(p+t) \times (p+t)}$ represents the pairwise kernel evaluation between the evaluated configurations. The hyperparameter $l$ controls how smoothly the function varies across the parameter space. Smaller values of $l$ make the kernel decay more rapidly with distance, allowing the function to capture more localized variations. In contrast, larger values of $l$ result in a more broader correlation across the parameter space, thereby enforcing smoother and lower variational structure. In practice, $l$ is chosen by maximizing the log marginal likelihood as
\begin{equation*}
    \mathcal{L}^t_{n}(l)=
    -\frac{1}{2} (\bar{\boldsymbol{f}}^t_{n})^\top \left[\mathbf{k}(\Theta^t,\Theta^t;l)\right]^{-1} \bar{\boldsymbol{f}}^t_{n} 
    - \frac{1}{2} \log \det \mathbf{k}(\Theta^t,\Theta^t;l) 
    - \frac{p+t}{2} \log(2\pi).
\end{equation*}

\subsection{Sequential Design and Acquisition Strategies}

As noted, metamodeling approaches are employed to reduce the computational cost of the time-consuming computer models. These surrogates are refined sequentially by selecting new parameter values in a way that exploits the surrogate’s structure.
This process is carried out by optimizing an \textit{acquisition function}, which provides the criterion for choosing the next evaluation point. When kriging is used, this sequential process is typically referred to as \textit{Bayesian optimization}, particularly in the machine learning domain \cite{frazier2018tutorial}. Some commonly used acquisition functions are Lower Confidence Bound (LCB), Probability of Improvement (PI), and Expected Improvement (EI) \cite{kleijnen2019simulation, sha2020applying}.

In the following, we review three commonly used acquisition functions. 
For consistency, we adopt the same notation as before. 

\subsubsection*{Lower Confidence Bound} 

Lower Confidence Bound (LCB) is a widely used acquisition function that balances exploration and exploitation when selecting the next evaluation point. In the minimization context, LCB favors points where the predicted mean is low with higher uncertainty (to explore). 
The LCB acquisition function is defined as:
\begin{equation*}
    \text{LCB}^t_{n}(\theta) = \mu^t_{n}(\theta) - \kappa \sigma^t_{n}(\theta)
\end{equation*}
where $\kappa > 0$ is a user-defined parameter that balances the exploitation and exploration. Larger values of $\kappa$ emphasize exploration by giving more weight to uncertainty, while smaller values trigger more exploitation over the mean prediction. In this setting, optimizing the acquisition function is carried out by minimizing the LCB, i.e., $\theta_{t+1} = \arg\min_{\theta \in \Theta}\text{LCB}^t_{n}(\theta)$.

\subsubsection*{Probability of Improvement} 

Probability of Improvement (PI) selects the next evaluation point based on the likelihood that it is expected to yield a lower value than the current best observed value so far. Define $v^t_{n}$ as the current best observed value among the evaluated configurations, $v^t_{n} := \min \bar{\boldsymbol{f}}^t_{n}$. The PI acquisition function is defined as:
\begin{equation*}
    \text{PI}^t_{n}(\theta) = \Phi\left( \frac{v^t_{n} - \mu^t_{n}(\theta)}{\sigma^t_{n}(\theta)} \right),
\end{equation*}
for the detailed derivation, see \Cref{siam:appendix:pi}. The next evaluation point is carried out by maximizing the PI, i.e., $\theta_{t+1} = \arg\max_{\theta \in \Theta}\text{PI}^t_{n}(\theta)$.

\subsubsection*{Expected Improvement} Expected Improvement (EI) selects the next evaluation point based on the expected magnitude of improvement over the current best observed value. Define the standardized threshold $z^t_{n}(\theta) := \frac{v^t_{n} - \mu^t_{n}(\theta)}{\sigma^t_{n}(\theta)}$. The EI acquisition function is defined as:
\begin{equation*}
    \text{EI}^t_{n}(\theta)
    = 
    (v^t_{n} - \mu^t_{n}(\theta)) \Phi(z^t_{n}(\theta)) 
    + \sigma^t_{n}(\theta) \phi(z^t_{n}(\theta)),
\end{equation*}
for the detailed derivation, see \Cref{siam:appendix:ei}. Using EI, the next evaluation point is obtained by maximizing EI, i.e., $\theta_{t+1} = \arg\max_{\theta \in \Theta} \text{EI}^t_{n}(\theta).$

\subsubsection*{Acquisition Function Optimization}  

In the sequential design framework, once the acquisition is chosen, the next evaluation point is decided by optimizing the acquisition function. One common approach for this step is to employ a multi-start strategy, where several initial points are randomly sampled across the parameter space and then refined using an efficient local optimization solver \cite{frazier2018tutorial}. 
Quasi-Newton or other nonlinear optimization methods are typically employed when derivatives can be computed efficiently. Each evaluation of the acquisition function requires computing the posterior predictive distribution at the candidate point, which involves inverting the regularized kernel matrix in \eqref{siam:eq:kriging-pred-dist}. This matrix inversion is usually the most computationally intensive part of the sequential design procedure. Accordingly, as more observations are accumulated, acquisition function optimization becomes more computationally expensive.

\section{Calibration with Root Finding Strategy}\label{siam:sec:rf-metamodeling}

In this section, we introduce an alternative root finding framework to solve \eqref{siam:eq:spaa-saa}. Define the sample root-finding objective as
\begin{equation}\label{siam:eq:spaarf-saa}
    \tilde{f}_n(\theta)
    :=
    \frac{1}{m_y}\frac{1}{n}\sum_{j=1}^n \boldsymbol{1}^\top r_j(\theta),
\end{equation}
and since $f_n(\theta)=(\tilde{f}_n(\theta))^2$, solving \eqref{siam:eq:spaa-saa} is equivalent to finding a root of \eqref{siam:eq:spaarf-saa} whenever such a root exists. The key motivation is that aggregating residuals before squaring preserves sign information, which can be effectively leveraged to guide the calibration process. Throughout, we assume that the sign of the discrepancy is treated symmetrically, without favoring either positive or negative values, although weighting schemes may be incorporated if desired to favor one over the other. We first establish the following assumptions.

\begin{assumption}[Regularity Conditions on the Sample Root-Finding Problem]
\label{siam:assump:reg:sample-rf}~
    
    \begin{enumerate}[label=(\alph*), ref=\theassumption(\alph*)]
    
    \item\label{siam:assump:existence} There exists a $\theta^*_n \in \Theta$ such that $\tilde{f}_{n}(\theta^*_n)=0$.
    
    \item\label{siam:assump:continuity} $\tilde{f}_{n}(\theta)$ is continuous 
    on $\theta \in \Theta$.
    
    \end{enumerate}
\end{assumption}

\begin{theorem}[Bolzano’s Theorem]\label{siam:thm:bolzano}
    Let $\tilde{f}: \Theta \to \mathbb{R}$ be continuous on a closed interval $[\theta_a, \theta_b] \subset \Theta$. If $\tilde{f}(\theta_a) \cdot \tilde{f}(\theta_b) < 0$, then there exists $\theta^* \in (\theta_a, \theta_b)$ such that $\tilde{f}(\theta^*) = 0$.
\end{theorem}

While \Cref{siam:assump:existence} does not hold in general for finite $n$, we first impose this for clarity and later discuss in \Cref{siam:sec:rootless} where no root exists. For \Cref{siam:assump:continuity}, this assumption is typically satisfied, since most commonly used kernels such as RBF or Matérn, are continuous, and such continuity assumptions appear frequently in many computer model applications \cite{ord1975computer, cellier2013continuous}.
\Cref{siam:assump:existence} and \Cref{siam:assump:continuity} are particularly important, since under \Cref{siam:thm:bolzano}, they ensure the existence of a root when the discrepancy function exhibits opposite signs at two distinct points. Importantly, this enables sign-guided calibration, in which we exploit the signs to progressively narrow down the search space. \Cref{siam:fig:metamodel-min-vs-rf} illustrates this structural advantage.

\begin{figure}[!htbp]
    \centering
    \includegraphics[width=0.70\linewidth]{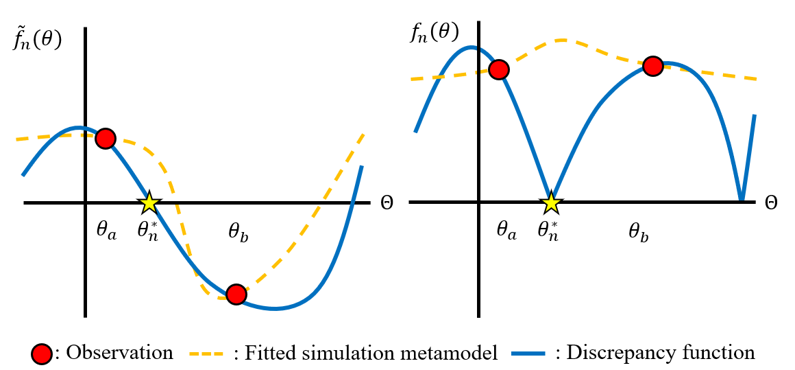}
    \caption{
        Left: a metamodel fitted to the signed discrepancy $\tilde{f}_n(\theta)$, where opposite signs at $\theta_a$ and $\theta_b$ guarantees the existence of a root within $[\theta_a, \theta_b]$ by continuity.  Right: a metamodel fitted to the squared discrepancy $f_n(\theta)$, where 
        additional evaluations may be needed to locate lower estimates within $[\theta_a, \theta_b]$.
    }
    \label{siam:fig:metamodel-min-vs-rf}
\end{figure}
Adapting this root finding strategy to the metamodeling framework, however, requires new acquisition functions, since the conventional acquisition functions are designed to target the regions with lower values. Instead, root finding acquisition function should target where the regions are expected to be close-to-zero. In what follows, we propose the root finding variants of the acquisition functions reviewed in \Cref{siam:sec:det-metamodeling}. 

\subsection{Proposed Root finding Acquisition Functions}

Following similar notation as before, to distinguish these quantities, we denote $\tilde{\eta}^t_{n}(\theta) \sim \mathcal{N}(\tilde{\mu}^t_{n}(\theta), (\tilde{\sigma}^t_{n})^2(\theta)),$
as the predictive distribution that is constructed with $\tilde{f}_n(\theta)$ from  \eqref{siam:eq:spaarf-saa}. Accordingly, in \eqref{siam:eq:kriging-pred-dist}, we use $\tilde{\boldsymbol{f}}^t_{n} = [\tilde{f}_{n}(\theta_1), \dots, \tilde{f}_{n}(\theta_{p+t})]^\top$ instead. 

\subsubsection*{Root finding with Lower Confidence Bound}

We extend the LCB acquisition function to root finding by favoring points where the predicted discrepancy is close-to-zero while encouraging exploration through uncertainty. Accordingly, we define the root finding LCB as:
\begin{equation*}
    {\widetilde{\text{LCB}}}^t_{n}(\theta) = |\tilde{\mu}^t_{n}(\theta)| - \kappa \tilde{\sigma}^t_{n}(\theta).
\end{equation*}
Since the predictive variance is always non-negative, only the mean component needs to be adjusted. As in the standard LCB, the balance between these exploitation and exploration is controlled by the user-defined parameter $\kappa > 0$. 
The next evaluation point is selected by minimizing ${\widetilde{\text{LCB}}}^t_{n}$, i.e., $\theta_{t+1} = \arg\min_{\theta \in \Theta} {\widetilde{\text{LCB}}}^t_{n}(\theta).$

\subsubsection*{Root finding with Probability of Improvement}

In the root finding setting, the improvement event is redefined to measure how much the predicted value moves closer to zero relative to the currently closest observed value, where ${i^t_{n}}^* := \arg\min_{1 \leq i \leq p+t} \left|[\tilde{\boldsymbol{f}}^t_{n}]_i\right|$ and $\tilde{v}^t_{n} := [\tilde{\boldsymbol{f}}^t_{n}]_{{i^t_{n}}^*}$ denote the index and evaluated configuration value whose observed discrepancy is closest to zero. The resulting root finding PI is defined as
\begin{equation*}
    {\widetilde{\text{PI}}}^t_{n}(\theta)
    = \Phi\left( \frac{|\tilde{v}^t_{n}| - \tilde{\mu}^t_{n}(\theta)}{\tilde{\sigma}^t_{n}(\theta)} \right) - \Phi\left( \frac{-|\tilde{v}^t_{n}| - \tilde{\mu}^t_{n}(\theta)}{\tilde{\sigma}^t_{n}(\theta)} \right),
\end{equation*}
for the detailed derivation, see \Cref{siam:appendix:pirf}.

\begin{figure}[!htbp]
    \centering
    \includegraphics[width=0.9\linewidth]{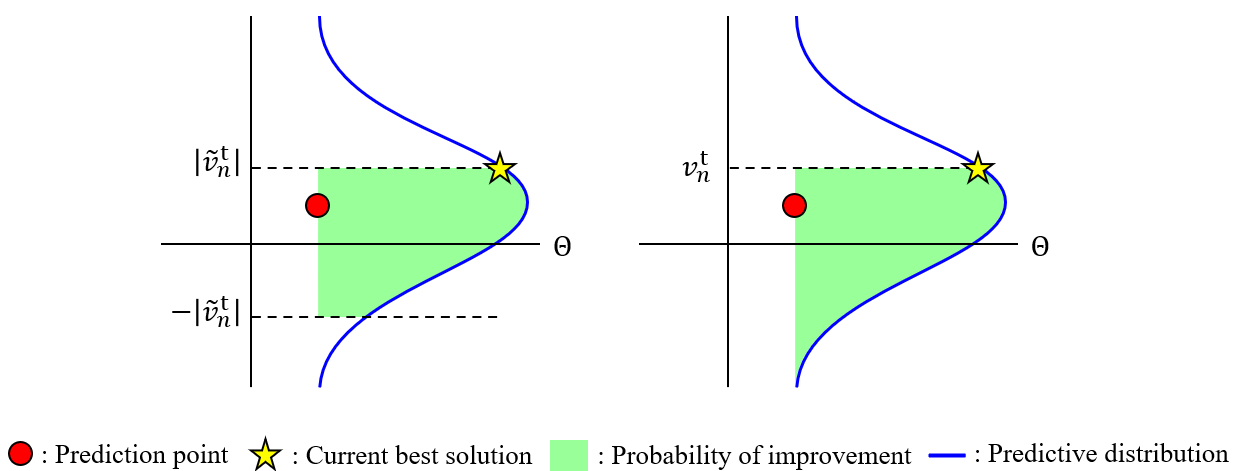}
        \caption{Probability of improvement under the root finding (left) and minimization framework (right).}
    \label{siam:fig:rf_improvement}
\end{figure}

\Cref{siam:fig:rf_improvement} illustrates the conceptual difference of improvement criteria between the minimization and root finding. In the minimization, improvement is defined as the predictive distribution producing a value lower than the current best value, resulting in a one-sided threshold. However, the root finding seeks predictions that are closer to zero, resulting in a two-sided threshold. The next evaluation point is then selected by maximizing ${\widetilde{\text{PI}}}^t_{n}$, i.e., $\theta_{t+1} = \arg\max_{\theta \in \Theta} {\widetilde{\text{PI}}}^t_{n}(\theta).$

\subsubsection*{Root finding with Expected Improvement}
In the root finding framework, the expected amount of improvement is redefined to quantify the expected reduction in the absolute deviation from zero. Let ${\ubar{z}}^t_{n}(\theta) := \frac{-|\tilde{v}^t_{n}| - \tilde{\mu}^t_{n}(\theta)}{\tilde{\sigma}^t_{n}(\theta)}$, $\tilde{z}^t_{n}(\theta) := \frac{-\tilde{\mu}^t_{n}(\theta)}{\tilde{\sigma}^t_{n}(\theta)}$, and ${\bar{z}}^t_{n}(\theta) := \frac{|\tilde{v}^t_{n}| - \tilde{\mu}^t_{n}(\theta)}{\tilde{\sigma}^t_{n}(\theta)}$ denote the standardized lower, central, and upper thresholds, respectively. The resulting root finding expected improvement is defined as
\begin{align*}
    {\widetilde{\text{EI}}}^t_{n}(\theta) 
    &= |\tilde{v}^t_{n}| \left( \Phi({\bar{z}}^t_{n}(\theta)) - \Phi({\ubar{z}}^t_{n}(\theta)) \right) 
    \\
    &
    \quad
    + 
    \tilde{\mu}^t_{n}(\theta) \left( 2\Phi({z}^t_{n}(\theta)) - \Phi({\bar{z}}^t_{n}(\theta)) - \Phi({\ubar{z}}^t_{n}(\theta)) \right) \\
    &
    \quad 
    - 
    \tilde{\sigma}^t_{n}(\theta) \left( 2\phi({z}^t_{n}(\theta)) - \phi({\bar{z}}^t_{n}(\theta)) - \phi({\ubar{z}}^t_{n}(\theta)) \right),
\end{align*}
for the detailed derivation, see \Cref{siam:appendix:eirf}. The next evaluation point is then selected by maximizing ${\widetilde{\text{EI}}}^t_{n}$, i.e., $\theta_{t+1} = \arg\max_{\theta \in \Theta} {\widetilde{\text{EI}}}^t_{n}(\theta).$

\subsection{Accelerating Acquisition Function Optimization}

Optimizing the acquisition function is the most computationally demanding step when using metamodel. 
In the root finding, however, this optimization step can be effectively accelerated by exploiting the sign information.
That is, we can bypass regions where no sign change occurs, thereby concentrating optimization efforts on those areas that are more likely to contain a root. Based on this insight, we introduce a search space reduction strategy tailored to this framework. We additionally derive the closed form first-order derivatives of the proposed root finding acquisition functions (see \Cref{siam:sec:appendix:gradient derivation} and \Cref{siam:sec:appendix:gradient validation}) to facilitate the use of gradient-based optimizer. 
\subsection*{Search Space Reduction}

As additional evaluations are acquired, the number of sign changing intervals (combinations of configurations) also increases. The key question is therefore which intervals to select and prioritize for subsequent evaluations. Since we assume no prior structural knowledge of the function, we focus on the most compact interval that contains a root. We quantify such compactness using a hypervolume measure based on the size of each sign-changing region in the parameter space. To formalize this, we first identify candidate subregions by examining all pairs $(\theta_a, \theta_b)$ with $1 \leq a < b \leq p+t$ such that $[\tilde{\boldsymbol{f}}^t_{n}]_a \cdot [\tilde{\boldsymbol{f}}^t_{n}]_b < 0$, 
where each pair defines a hyperrectangular subregion $R(\theta_a, \theta_b)$, 
\begin{equation*}
    R(\theta_a, \theta_b) = \prod\nolimits_{l=1}^{m_\theta} 
    \left[\min\{[\theta_{a}]_l, [\theta_{b}]_l\}, \max\{[\theta_{a}]_l, [\theta_{b}]_l\}\right].
\end{equation*}
The hypervolume of each subregion is computed as 
\begin{equation*}
    V(\theta_a, \theta_b) = \left(\prod\nolimits_{l=1}^{m_\theta} \max\{\left|[\theta_{a}]_l - [\theta_{b}]_l\right|, \underline{\theta}\}\right) 
    \cdot 
    \left|[\tilde{\boldsymbol{f}}^t_{n}]_a - [\tilde{\boldsymbol{f}}^t_{n}]_b\right|,
\end{equation*}
where $\underline{\theta}>0$ is a user-defined distance threshold to prevent zero-volume subregions when coordinates coincide along any axis. Figure~\ref{siam:fig:rss} illustrates this process and overall algorithmic framework is presented in \Cref{siam:alg:rss}.

\begin{algorithm}[htbp]
\caption{Reduced Search Space (RSS) - deterministic}
\label{siam:alg:rss}
\begin{algorithmic}[1]
\STATE \textbf{Initialization:} Evaluated configurations $\Theta^t = \{\theta_1, \dots, \theta_{p+t}\} \in \mathbb{R}^{(p+t) \times m_\theta}$; 
corresponding signed discrepancies $\tilde{\boldsymbol{f}}^t_{n} = [\tilde{f}_{n}(\theta_1), \dots, \tilde{f}_{n}(\theta_{p+t})]^\top$;
distance threshold $\underline{\theta} > 0$
\STATE $\mathcal{S} \gets \emptyset$ to store candidate subregions
\FOR{each pair $(a, b)$ with $1 \leq a < b \leq p+t$}
    \IF{$\tilde{f}_{n}(\theta_a) \cdot \tilde{f}_{n}(\theta_b) < 0$}
        \STATE Define subregion:
        $R(\theta_a, \theta_b) = \prod_{l=1}^{m_\theta} \left[\min\{[\theta_{a}]_l, [\theta_{b}]_l\}, \max\{[\theta_{a}]_l, [\theta_{b}]_l\}\right]$
        \STATE Compute hypervolume: 
            $V(\theta_a, \theta_b) = \left(\prod_{l=1}^{m_\theta} \max\{\left|[\theta_{a}]_l - [\theta_{b}]_l\right|, \underline{\theta}\}\right) 
            \cdot 
            \left|[\tilde{\boldsymbol{f}}^t_{n}]_a - [\tilde{\boldsymbol{f}}^t_{n}]_b\right|$
        \STATE Append $(R(\theta_a, \theta_b), V(\theta_a, \theta_b))$ to $\mathcal{S}$
    \ENDIF
\ENDFOR
\STATE Select subregions in $\mathcal{S}$ with the smallest $V(\theta_a, \theta_b)$
\end{algorithmic}
\end{algorithm}

\begin{figure}[htbp]
    \centering
    \includegraphics[width=0.45\linewidth]{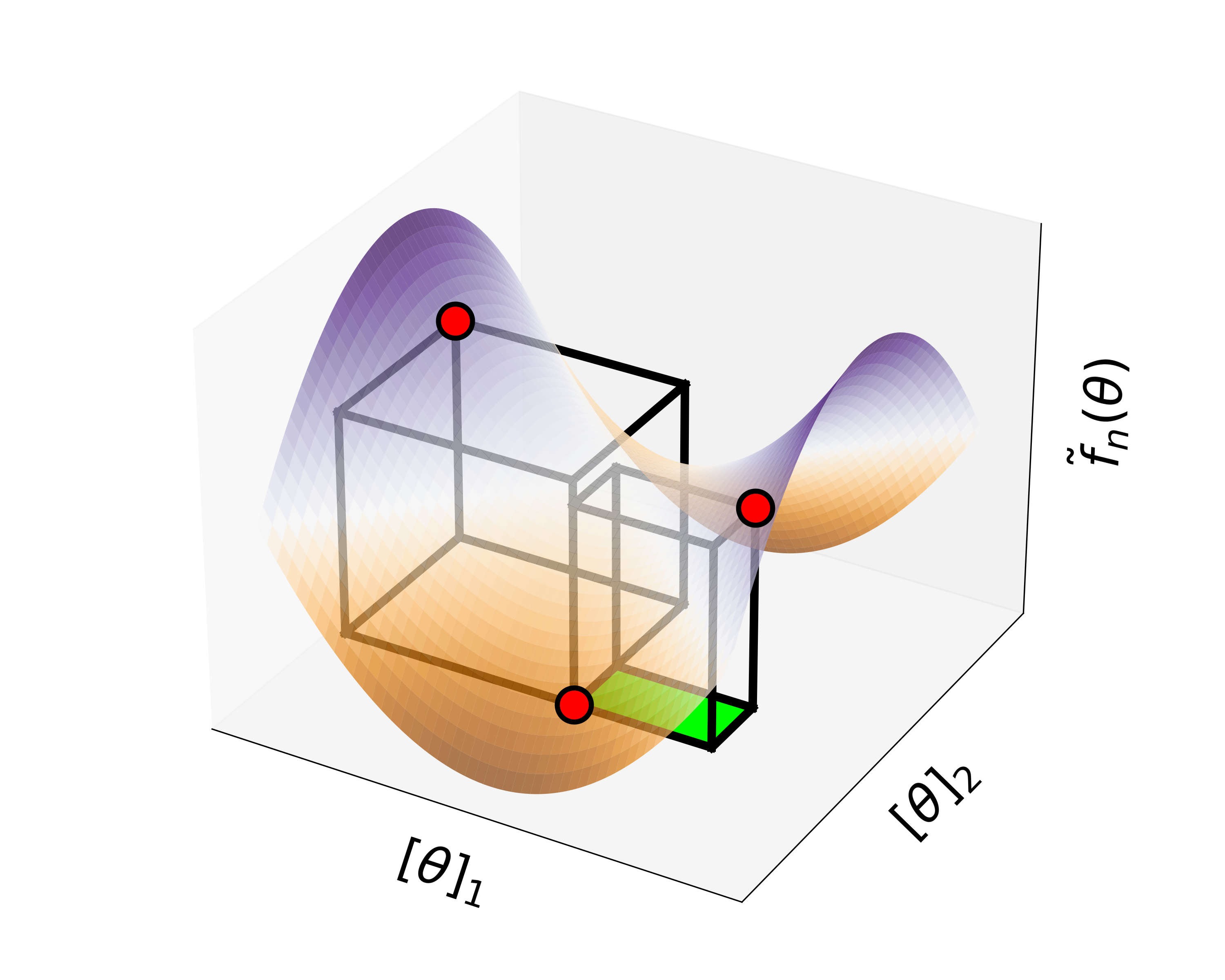}
    \caption{Computed hypervolume from three evaluated configurations in a two-dimensional input space ($m_\theta = 2$), among which two pairs exhibit opposite signs. These pairs define two candidate subregions, and the subregion with the smallest hypervolume is highlighted in green.}
    \label{siam:fig:rss}
\end{figure}

\subsection{Root Finding in Rootless Scenarios}
Throughout, the core of this section was motivated by the sign information under \Cref{siam:assump:existence} and \Cref{siam:assump:continuity}. 
However, \Cref{siam:assump:existence} need not hold for finite $n$, and $\tilde f_n(\theta)$ may remain strictly positive (or negative) over $\Theta$. In practice, the functional structure is unknown in advance, so it is unclear whether one should abandon the root finding formulation and refit the surrogate to the minimization objective $\bar f_n(\theta)$ based on the observed residuals before collecting any excessive evaluations. We refer to this as the \textit{rootless} scenario, and note that RSS is also inapplicable here since it relies on sign-changing intervals. 
We discuss this scenario on a realized sample path of $\tilde{f}_n$. For this, let $(\Omega, \mathcal{F}, \mathbb{P})$ be the probability space governing the randomness in $\mathcal{D}_n$, and define $\Omega_+ := \{\omega \in \Omega : \tilde{f}_n(\theta; \omega) > 0 \text{ for all } \theta \in \Theta\}$ as the event corresponding to the rootless scenario. We then establish the following result.
\begin{theorem}[Asymptotic Equivalence in the Rootless Scenario]\label{siam:thm:rootless-summary}
    For a fixed $\omega \in \Omega_+$, the optimizers of the root finding and standard acquisition functions coincide asymptotically as $t \to \infty$,
    \begin{align*}
        \left\|\arg\min_{\theta \in \Theta} \widetilde{\mathrm{LCB}}^t(\theta; \omega) - \arg\min_{\theta \in \Theta} \mathrm{LCB}^t(\theta; \omega)\right\| &\to 0, \\
        \left\|\arg\max_{\theta \in \Theta} \widetilde{\mathrm{PI}}^t(\theta; \omega) - \arg\max_{\theta \in \Theta} \mathrm{PI}^t(\theta; \omega)\right\| &\to 0, \\
        \left\|\arg\max_{\theta \in \Theta} \widetilde{\mathrm{EI}}^t(\theta; \omega) - \arg\max_{\theta \in \Theta} \mathrm{EI}^t(\theta; \omega)\right\| &\to 0.
    \end{align*}
\end{theorem}
For the detailed assumptions and proofs, see \Cref{siam:sec:rootless}. \Cref{siam:thm:rootless-summary} implies that the proposed acquisition functions at least asymptotically recovers the minimizer of $f_n$, where the gap between $f_n$ and the original objective $\bar{f}$ is characterized in \Cref{siam:sec:sol-acc-rf}.

\section{Calibration with Stochastic Simulation Metamodeling}\label{siam:sec:sto-metamodeling}

The deterministic approach in \Cref{siam:sec:det-metamodeling} relies on a fixed observed dataset $\mathcal{D}_n$. However, this finite-sample formulation has several limitations. First, it is subject to sampling variability and noise. Second, when $n$ is large, evaluating \eqref{siam:eq:p} becomes computationally expensive, as each calibration iteration requires processing the full dataset. Third, this sample average approximation implicitly assumes homoskedastic noise, which may not generally hold in practice. Stochastic kriging addresses these limitations by explicitly modeling the heteroskedastic nature of the observation. In this framework, the stochastic input-output pairs $(X,Y)$ can be obtained either through resampling methods (e.g., bootstrap) from $\mathcal{D}_n$, or by fitting a generative model. 
In what follows, we first review the standard minimization based stochastic kriging. 

\subsection{Stochastic Kriging}
Stochastic kriging models the mean response $\bar{f}(\theta)$ in \eqref{siam:eq:sp} as a realization of GRF. Unlike the deterministic setting in \Cref{siam:sec:det-metamodeling}, where residuals $r_j(\theta)$ were computed from the fixed observed $\mathcal{D}_n$, in stochastic kriging, each configuration $\theta$ is evaluated with $n(\theta)$ independent replications drawn from the underlying distribution $(X,Y)$, where $r_{(j)}(\theta) = y_{(j)} - h(x_{(j)}; \theta)$. We use separate notation for these residuals to emphasize that they arise from generated simulation draws at each $\theta$, rather than from a pre-existing fixed dataset. The replication average at $\theta$ is computed as
\begin{equation}\label{siam:eq:p-sto}
    \bar{f}_{n(\theta)}(\theta) 
    := 
    \frac{1}{m_y}\frac{1}{n(\theta)}\sum_{j=1}^{n(\theta)} \|r_{(j)}(\theta)\|_2^2
    .
\end{equation}
Let $\Theta^t = \{\theta_1, \ldots, \theta_{p+t}\}$ denote the set of evaluated configurations up to iteration $t$, and $\bar{\boldsymbol{f}}^t_n = [\bar{f}_{n(\theta_1)}(\theta_1), \ldots, \bar{f}_{n(\theta_{p+t})}(\theta_{p+t})]^\top$ be the vector of replication averages. These observations are modeled as
\begin{equation*}
    \bar{\boldsymbol{f}}^t_n = \bar{\boldsymbol{f}}^t + \boldsymbol{\varepsilon}^t, 
    \quad 
    \boldsymbol{\varepsilon}^t \sim \mathcal{N}(0, \Sigma^{(t)}),
\end{equation*}
where $\bar{\boldsymbol{f}}^t = [\bar{f}(\theta_1), \dots, \bar{f}(\theta_{p+t})]^\top$ denotes the vector of latent mean responses and $\boldsymbol{\varepsilon}^t$ represents the simulation noise. The covariance matrix $\Sigma^{(t)} \in \mathbb{R}^{(p+t) \times (p+t)}$ characterizes the noise structure induced by finite replications. Its diagonal entries quantify the sampling variance at each design point, while the off-diagonal entries capture potential correlations arising from common random numbers or other shared sources of randomness. The $a$-th diagonal entry of $\Sigma^{(t)}$ is estimated as 
\begin{equation*}
    [\Sigma^{(t)}]_{a,a} = \frac{1}{n(\theta_a)(n(\theta_a)-1)}
    \sum_{j=1}^{n(\theta_a)} \left( \frac{1}{m_y}\|r_{(j)}(\theta_a)\|_2^2 
    - \bar{f}_{n(\theta_a)}(\theta_a) \right)^2.
\end{equation*}
The replication size $n(\theta)$ governs the extent to which simulation noise is reduced at each design point. Increasing $n(\theta)$ decreases the variance of the replication average and gives more accurate estimates of the mean response. However, larger $n(\theta)$ also incur higher computational cost, so it is important to balance the replication effort and modeling accuracy. 
One approach for determining $n(\theta)$ is based on the integrated mean squared error (IMSE) criterion \cite{ankenman2010stochastic}. In this approach, replication sizes are allocated to minimize the expected mean squared prediction error integrated over the domain of interest. Assuming that the extrinsic covariance and intrinsic variance functions are known, the replication effort can be determined by quantifying its marginal impact on the IMSE score. Alternatively, $n(\theta)$ can be determined by balancing the exploration–exploitation trade-off between replicating at the visited points and evaluating new design points \cite{Sürer01082025}. Here, replication sizes are chosen by minimizing the integrated variance criterion, which is computed by comparing the posterior variance reduction achieved by adding replications at existing design points and the reduction achieved by introducing new design points sampled from the posterior distribution.

We now describe the predictive distribution of stochastic kriging. Given $\Theta^t$ and corresponding vector of replication averages $\bar{\boldsymbol{f}}^t_n$, we seek a predictor of the mean response $\bar{f}(\theta)$ at an unobserved configuration $\theta$. The predictive distribution of stochastic kriging is given by
\begin{align}\label{siam:eq:stochastickriging-pred-dist}
    \eta^t_n(\theta) &\sim \mathcal{N}\left(\mu^t_n(\theta), (\sigma^t_n)^2(\theta)\right),
    \\
    \mu^t_n(\theta) 
    &= 
    \mathbf{k}(\theta, \Theta^t;l) \left[\mathbf{k}(\Theta^t,\Theta^t;l) + \Sigma^{(t)}\right]^{-1} \bar{\boldsymbol{f}}^t_n
    \nonumber
    ,
    \\
    (\sigma^t_n)^2(\theta) 
    &= 
    k(\theta, \theta;l) 
    - \mathbf{k}(\theta,\Theta^t;l) \left[\mathbf{k}(\Theta^t,\Theta^t;l) + \Sigma^{(t)}\right]^{-1} 
    \mathbf{k}(\theta,\Theta^t;l)^\top,
    \nonumber
\end{align}
for the detailed derivation, see \Cref{siam:sec:SK-pred-derivation}. These expressions closely mirror the standard kriging predictor in \eqref{siam:eq:kriging-pred-dist}, where the difference is that the simulation noise covariance $\Sigma^{(t)}$ is explicitly incorporated. The hyperparameter $l$ is estimated by maximizing the log marginal likelihood
\begin{equation}\label{siam:eq:stochastickriging-mle}
    \mathcal{L}^t_n(l)=
    -\frac{1}{2} (\bar{\boldsymbol{f}}_n^t)^\top \left[\mathbf{k}(\Theta^t,\Theta^t;l) + \Sigma^{(t)}\right]^{-1} \bar{\boldsymbol{f}}_n^t
    - \frac{1}{2} \log \det \left(\mathbf{k}(\Theta^t,\Theta^t;l) + \Sigma^{(t)}\right)
    - \frac{p+t}{2} \log(2\pi).
\end{equation}

\subsection{Sequential Design and Acquisition Strategies}

The sequential design strategies for stochastic kriging follow the same general principles as in \Cref{siam:sec:det-metamodeling}. A conventional approach is to directly apply the standard acquisition functions in \Cref{siam:sec:det-metamodeling},
under the assumption that the stochasticity is implicitly accounted through $\Sigma^{(t)}$. Here, alternative resampling methods, such as jackknife-based methods, can also be used to refine the variance estimation \cite{tang2022borderline}. Moreover, several refined acquisition functions have also been proposed to better accommodate the stochastic settings.
One example is the Augmented Expected Improvement (AEI), which modifies EI by penalizing repeated sampling at the current best configuration \cite{huang2006global}. However, AEI relies on the assumption that the noise variance is known, which is not practically satisfied in many applications \cite{amini2025tutorial}. Another related work is Modified Expected Improvement (MEI), which is proposed within a two-stage framework. In its initial stage, MEI accounts only for spatial prediction uncertainty, while in the second stage, replication effort is distributed based on the estimated noise variance \cite{quan2013simulation}. Furthermore, several choices have been discussed in the literature for defining the current best value used in PI and EI, such as using the predictive mean \cite{zhou2024corrected} or quantile-based approach \cite{picheny2013quantile}. In this work, we define the current best value as the minimum predictive mean among $\Theta^t$,
\begin{equation}\label{siam:eq:recommend-rule}
    {i^t_{n}}^* := \arg\min_{1 \leq i \leq p+t}, \mu^t_{n}(\theta_i),
    \quad
    v^t_{n} := \mu^t_{n}(\theta_{{i^t_{n}}^*}).
\end{equation}

\subsection{Root finding Modifications}

We now extend the proposed root finding framework to the stochastic setting. Here, observations are no longer noise-free, and the presence of stochasticity implies that all quantities of interest, such as the sign and the existence of a root, are probabilistic in nature. We first write the stochastic root finding observation as
\begin{equation*}
    \tilde{f}_{n(\theta)}(\theta) 
    := \frac{1}{m_y}\frac{1}{n(\theta)}\sum_{j=1}^{n(\theta)} \boldsymbol{1}^\top r_{(j)}(\theta),
\end{equation*}
which mirrors the expression in \eqref{siam:eq:p-sto}.
Given the set of design points $\Theta^t$, the vector of collected sample averages is $\tilde{\boldsymbol{f}}^t_n = [\tilde{f}_{n(\theta_1)}(\theta_1), \ldots, \tilde{f}_{n(\theta_{p+t})}(\theta_{p+t})]^\top$.
The $a$-th diagonal entries for simulation-noise covariance matrix $\tilde{\Sigma}^{(t)} \in \mathbb{R}^{(p+t)\times(p+t)}$ is
\begin{equation*}
    [\tilde{\Sigma}^{(t)}]_{a,a}
    =
    \frac{1}{n(\theta_a)(n(\theta_a)-1)}
    \sum_{j=1}^{n(\theta_a)}
    \left(
    \frac{1}{m_y}\boldsymbol 1^\top r_{(j)}(\theta_a)
    -
    \tilde f_{n(\theta_a)}(\theta_a)
    \right)^2.
\end{equation*}
All remaining modeling assumptions follow the convention in \eqref{siam:eq:stochastickriging-pred-dist} and \eqref{siam:eq:stochastickriging-mle}, with $\bar{\boldsymbol{f}}^t_n$ and $\Sigma^{(t)}$ replaced by $\tilde{\boldsymbol{f}}^t_n$ and $\tilde{\Sigma}^{(t)}$, and denote the root finding predictive distribution as $\tilde{\eta}^t_n(\theta) \sim \mathcal{N}\left(\tilde{\mu}^t_n(\theta), (\tilde{\sigma}^t_n)^2(\theta)\right)$.

We now discuss the stochastic estimator that motivates \Cref{siam:assump:better-mse}. In this setting, we define the stochastic estimator in \Cref{siam:assump:better-mse} as
\begin{equation}\label{siam:eq:sto-min-define}
    \check{f}_n(\theta)
    :=
    \mathbb{E}\left[\left(\tilde{\eta}_n^t(\theta)\right)^2 \mid \mathcal{D}_n\right]
    =
    \left(\tilde{\mu}_n^t(\theta)\right)^2
    +(\tilde{\sigma}_n^t)^2(\theta).
\end{equation}
Under the GRF assumption, $\tilde{\eta}_n^t(\theta) | \mathcal{D}_n$ represents the posterior distribution of $\tilde{f}(\theta)$, so that 
\begin{equation*}
    \mathbb{E}\left[\left(\tilde{\eta}_n^t(\theta)\right)^2 \mid \mathcal{D}_n\right]
    =
    \mathbb{E}\left[\left(\tilde{f}(\theta)\right)^2 \mid \mathcal{D}_n\right]
    =
    \mathbb{E}\left[f(\theta) \mid \mathcal{D}_n\right].
\end{equation*}
As this becomes the posterior mean of $f(\theta)$ under $\mathcal{D}_n$, \eqref{siam:eq:sto-min-define} achieves the minimum mean squared error among $\Theta$. Therefore, for each fixed $\theta \in \Theta$,
\begin{equation}\label{siam:eq:better-mse-point}
    \mathbb{E}\left[
        \left(
            \check{f}_n(\theta) - f(\theta)
        \right)^2
    \right]
    \leq
    \mathbb{E}\left[
        \left(
            f_n(\theta) - f(\theta)
        \right)^2
    \right].
\end{equation}
Although \eqref{siam:eq:better-mse-point} does not by itself imply \Cref{siam:assump:better-mse}, this implies that for any $\theta \in \Theta$, $\check{f}_n$ achieves a smaller mean squared error than $f_n$. Now, given that $\Theta$ is compact, if this pointwise error improvement is sufficiently uniform across $\Theta$, or when the kernel is bounded and sufficiently smooth so that $(\check{f}_n(\theta) - f(\theta))^2$ is relatively stable across $\Theta$, then \Cref{siam:assump:better-mse} becomes more plausible. Moreover, when the design points are well dispersed so that the resulting predictive variance does not drastically changes over $\Theta$, this would further justify \Cref{siam:assump:better-mse}. Accordingly, the refined solution is obtained by minimizing $\check f_n(\theta)$ over $\Theta$, where we define the current best value in the root-finding as
\begin{equation}\label{siam:eq:recommend-rule-rf}
    {i^t_{n}}^* := \arg\min_{1 \leq i \leq p+t}
    \left\{
    (\tilde{\mu}^t_{n}(\theta_i))^2 + (\tilde{\sigma}^t_{n})^2(\theta_i)
    \right\},
    \quad
    \tilde{v}^t_{n} := \tilde{\mu}^t_{n}(\theta_{{i^t_{n}}^*}).
\end{equation}

\subsubsection*{Search Space Reduction}

In the stochastic setting, \Cref{siam:thm:bolzano} no longer applies directly since we can only observe noisy sample averages. 
For this, we instead quantify the probability of root existence using the stochastic kriging posterior,
\begin{equation}\label{siam:eq:sign-change-prob}
    \mathbb{P}(\mathcal{I}^{\pm}(\theta_a, \theta_b)) 
    = 
    \Phi\left(-\frac{\tilde{\mu}^t_n(\theta_a)}{\tilde{\sigma}^t_n(\theta_a)}\right) 
    + \Phi\left(-\frac{\tilde{\mu}^t_n(\theta_b)}{\tilde{\sigma}^t_n(\theta_b)}\right) 
    - 2\Phi\left(-\frac{\tilde{\mu}^t_n(\theta_a)}{\tilde{\sigma}^t_n(\theta_a)}\right)
    \Phi\left(-\frac{\tilde{\mu}^t_n(\theta_b)}{\tilde{\sigma}^t_n(\theta_b)}\right),
\end{equation}
where $\mathcal{I}^{\pm}(\theta_a, \theta_b)$ denotes the event that $\tilde{f}_n(\theta_a)$ and $\tilde{f}_n(\theta_b)$ have opposite signs (see \Cref{siam:sec:prob-rootexist}).

In the deterministic RSS, the search process was triggered when the evaluated points exhibited opposite signs. However, in the stochastic setting, this triggering condition has to be modified since every pair of configurations has a non-zero probability of having opposite signs. For this, we compute the probability of root existence for all configuration pairs and narrow the search only to those that are likely to contain a root with high probability, i.e., $\mathbb{P}(\mathcal{I}^{\pm}(\theta_a, \theta_b)) \geq \alpha$, where $\alpha \in [0, 1]$ is a user-defined probability threshold. Accordingly, we define the hypervolume measure as
\begin{equation*}
    V(\theta_a, \theta_b) 
    = \left(\prod\nolimits_{l=1}^{m_\theta} \max\{|[\theta_a]_l - [\theta_b]_l|, \underline{\theta}\}\right) 
    \cdot (1 - \mathbb{P}(\mathcal{I}^{\pm}(\theta_a, \theta_b))),
\end{equation*}
where the goal is to identify a compact region with a smaller hypervolume. \Cref{siam:fig:rss_stochastic} illustrates this process and  detailed algorithmic framework is presented in  \Cref{siam:alg:rss-stochastic}.

\begin{figure}[htbp]
    \centering
    \includegraphics[width=0.45\linewidth]{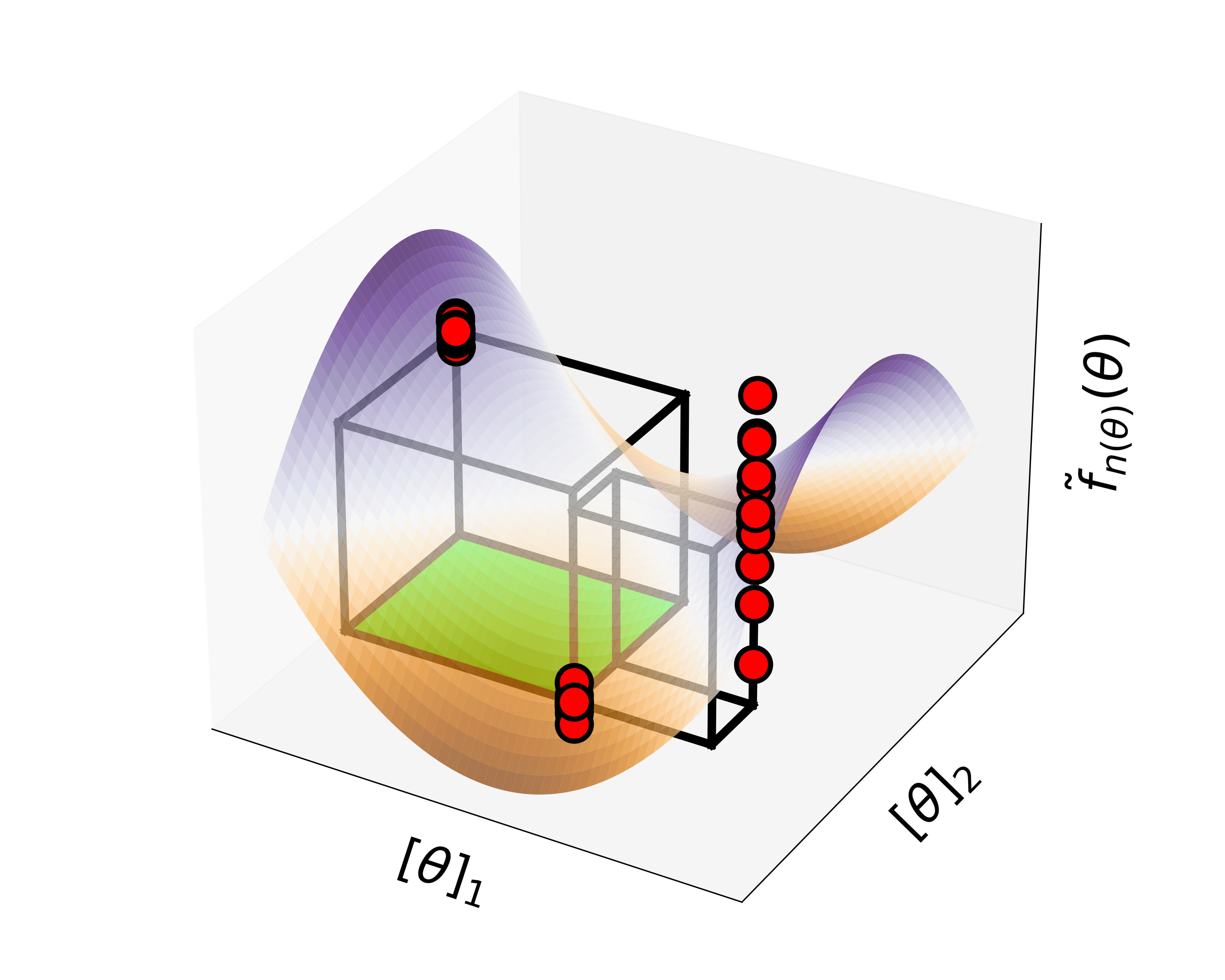}
    \caption{RSS in a stochastic setting in a two-dimensional input space ($m_\theta = 2$). Under the probabilistic framework, the hypervolume accounts for sampling variability, resulting in a different selected subregion compared to 
    \Cref{siam:fig:rss}.}
\label{siam:fig:rss_stochastic}
\end{figure}

\begin{algorithm}[htbp]
\caption{Reduced Search Space (RSS) - stochastic}
\label{siam:alg:rss-stochastic}
\begin{algorithmic}[1]
\STATE \textbf{Initialization:} Evaluated configurations $\Theta^t = \{\theta_1, \ldots, \theta_{p+t}\} \in \mathbb{R}^{(p+t) \times m_\theta}$; predictive means $\tilde{\mu}^t_n(\theta)$ and standard deviations $\tilde{\sigma}^t_n(\theta)$ for all $\theta \in \Theta^t$; probability threshold $\alpha \in [0, 1]$; distance threshold $\underline{\theta} > 0$
\STATE $\mathcal{S} \gets \emptyset$ to store candidate subregions
\FOR{each pair $(a, b)$ with $1 \leq a < b \leq p+t$}
    \IF{$\mathbb{P}(\mathcal{I}^{\pm}(\theta_a, \theta_b)) \geq \alpha$}
        \STATE Define subregion: $R(\theta_a, \theta_b) = \prod_{l=1}^{m_\theta}[\min\{[\theta_a]_l, [\theta_b]_l\}, \max\{[\theta_a]_l, [\theta_b]_l\}]$
        \STATE Compute hypervolume:
        $V(\theta_a, \theta_b) = \left(\prod_{l=1}^{m_\theta} \max\{|[\theta_a]_l - [\theta_b]_l|, \underline{\theta}\}\right) \cdot (1 - \mathbb{P}(\mathcal{I}^{\pm}(\theta_a, \theta_b)))$
        
        \STATE Append $(R(\theta_a, \theta_b), V(\theta_a, \theta_b))$ to $\mathcal{S}$
    \ENDIF
\ENDFOR
\STATE Select subregions in $\mathcal{S}$ with the smallest $V(\theta_a, \theta_b)$
\end{algorithmic}
\end{algorithm}

Here, the triggering condition in RSS may also be used to determine the replication size $n(\theta)$. 
Since larger replication sizes reduce the predictive variance (uncertainty) by reducing the simulation noise, replication can be allocated in two complementary ways. First, within a fixed replication budget at each iteration, additional replications may be assigned to reduce uncertainty at evaluated configurations where no candidate pair has yet satisfied the triggering condition $\alpha$. Secondly, once the potential pairs are identified, replications can be allocated to those configurations under the ranking and selection method \cite{kim2007recent}, continuing until a desired level of statistical significance is achieved.

\section{Numerical Experiment}\label{siam:sec:experiment}
In this section, we conduct a range of experiments to evaluate the calibration performance of the proposed root finding framework against standard minimization. 
We compare three acquisition strategies reviewed in \Cref{siam:sec:det-metamodeling} with their root finding variants in \Cref{siam:sec:rf-metamodeling}. 
We set $\kappa = 1$ for both LCBs and $\alpha = 0.95$ for the stochastic RSS, and set $\underline{\theta} = 1e-8$. 
Each experiment consists of 100 independent macro replications. In each macro replication, we start from initial design of $p = 2$ points and perform up to 10 sequential evaluations. At each design point $\theta$, we conduct 10 independent replications to estimate the discrepancy. 
To optimize the acquisition function, we use a multi-start strategy with L-BFGS-B initialized from 10 randomly generated starting points, and each run is limited to 10 iterations.
To assess the calibration performance, we conduct a post-evaluation step \cite{eckman2023diagnostic} by estimating $\bar{f}$ at each current best solution using 1,000 independent replications. 

\subsection{2D Synthetic Example}
We consider a synthetic example where the computer model has parameter $\theta \in [-3, 3]^2$. The signed discrepancy is a modified Himmelblau function \cite{jeon2024calibrating},
\begin{equation*}
    \tilde{f}(\theta) 
    := \log_2 \left(\left([\theta]_1^2 + [\theta]_2 - 3\right)^2 + \left( [\theta]_1 + [\theta]_2^2 - 2\right)^2\right) - 1,
\end{equation*}
where each replication at $\theta$ is observed through
\begin{equation*}
    R_{(i)}(\theta) = \tilde{f}(\theta) + \varepsilon_{(i)}, 
    \quad \varepsilon_{(i)} \sim \mathcal{N}\big(0, \sigma^2(\theta)\big), 
    \quad \sigma(\theta) = \sqrt{|\tilde{f}(\theta)|},
\end{equation*}
in which the noise variance is heteroskedastic and increases with the magnitude of the discrepancy. \Cref{siam:fig:2d_func} depicts the signed discrepancy surface.
\begin{figure}[htbp]
    \centering
    \includegraphics[width=0.85\linewidth]{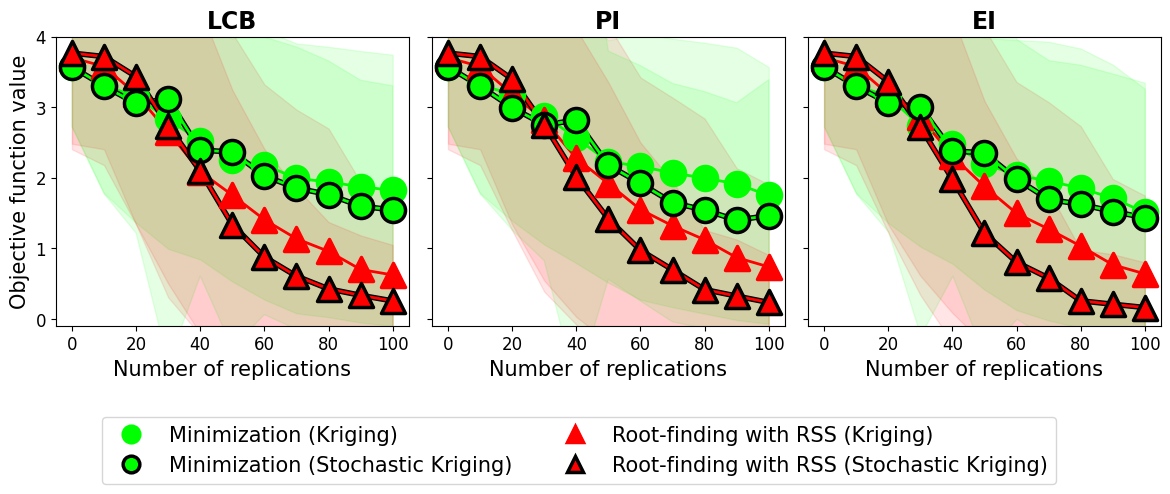}
    \caption{Calibration results of the 2D synthetic example. Objective function value at each  best observed solution is obtained from 1,000 post replications and reported with 95\% confidence intervals.}
    \label{siam:fig:2d_result}
\end{figure}

The experimental results are summarized in \Cref{siam:fig:2d_result}. Across all three acquisition functions, root finding with RSS consistently achieves faster convergence than minimization. Notably, stochastic kriging tends to offer additional improvement in both minimization and root-finding. However, this benefit is less reflected in minimization, as minimization requires more evaluations to identify promising region. During the early stage, evaluations are mostly exploratory, so the advantage of stochastic kriging is not exploited as effectively. In contrast, root finding offers a structural advantage that allows the modeling benefit to be leveraged more effectively. This result is also consistent with the analysis in \Cref{siam:sec:sol-acc-rf}, where stochastic kriging provides tighter solution bounds by accounting for the finite sample uncertainty. 

\begin{figure}[htbp]
    \centering    \includegraphics[width=1\linewidth]{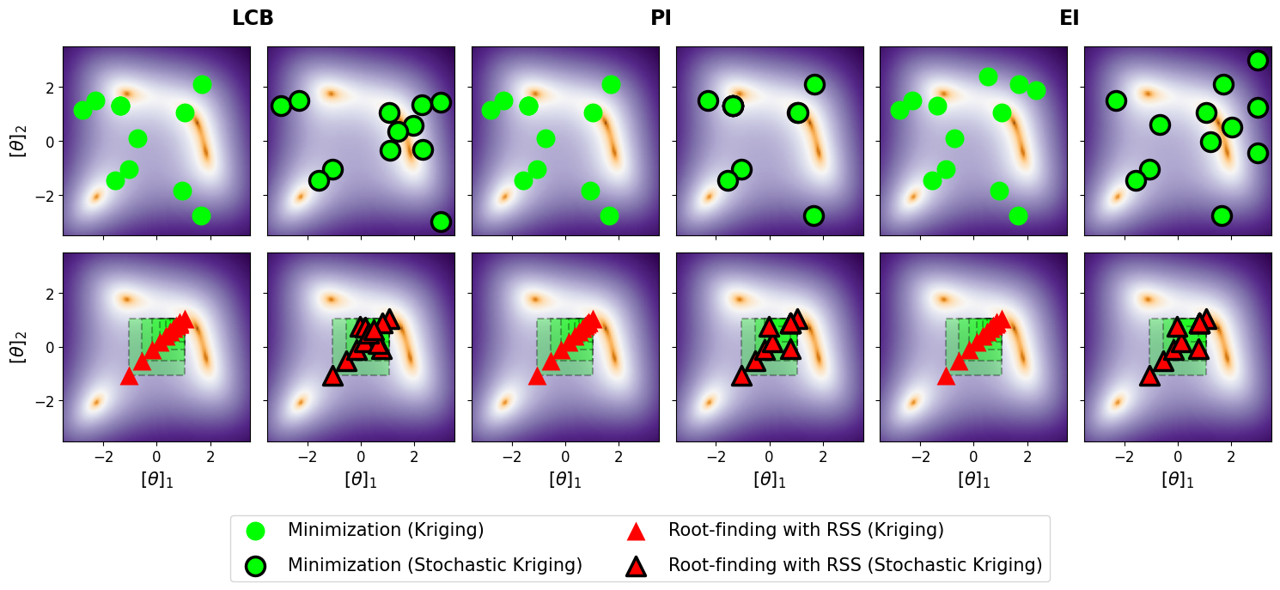}
    \caption{Sampling trajectories of the benchmark methods after 10 evaluations from a single macro replication. Dashed rectangles show the progressively reduced search space using RSS.}
    \label{siam:fig:2d_sampling_jrace}
\end{figure}

To better illustrate this, \Cref{siam:fig:2d_sampling_jrace} shows the sampling trajectories from a single macro replication after 10 evaluations. For the minimization, the sampling trajectory is more scattered and exploratory, as there is no structural guidance. In contrast, root finding produces a more directed trajectory by leveraging sign information. Particularly, in root finding, the sampling trajectories under kriging and stochastic kriging are noticeably different. In the deterministic case, the search space is progressively reduced toward a single region as the root is expected to exist as a fixed point. In the stochastic case, however, the probability of sign changes (from \eqref{siam:eq:sign-change-prob}) is adaptively updated as new observations are acquired, thereby the search space can switch across different region. This results in the sampling trajectories that remain close to the zeroth contour while being scattered around it.
 
\subsection{Logistic Simulation Example}

In this experiment, we consider a logistic simulation that resembles a classical M/M/1  system. The service rate $\mu^{\text{real}} = 4$ is assumed to be known and fixed, while the actual arrival rate $\theta^{\text{real}} = 6$ is unknown. In this system, entities (e.g., customers or jobs) arrive according to a Poisson process with rate $\theta^{\text{real}}$, and are served sequentially with exponentially distributed service times. We assume that only the sojourn times of entities are observable in the system, and thereby $\theta^{\text{real}}$ can only be estimated through computer simulation by comparing the simulated output with the observed data.

Such settings are commonly appeared in practice, where detailed system dynamics are often unavailable and thereby parameter estimation must rely on partially observable dataset \cite{glynn1989indirect, asanjarani2021survey}. Moreover, even if all arrival and service events were fully observable, estimating the arrival rate would still be challenging, as most of those tractable closed-form equations are derived under idealized assumptions, such as infinite-horizon or steady-state assumptions, which are not applicable in our finite entity setting. Furthermore, we do not attempt to restrict the problem to so called “well-behaved” setting in which the arrival rate is strictly less than the service rate (i.e., $\theta^{\text{real}} < \mu^{\text{real}}$). Instead, we allow the calibration task to directly reflect the observed system output as is, without imposing additional assumptions. For these reasons, the problem is addressed within a simulation based calibration framework. \Cref{siam:fig:mm1} illustrates the calibration setup with $m_y = 100$ observed entities.

\Cref{siam:fig:mm1_func} illustrates the estimated discrepancy functions along with their 95\% confidence intervals. As $\theta$ increases, the signed discrepancy becomes more negative, since a higher arrival rate makes the system more congested and leads to longer average sojourn times. The variance of the estimated discrepancy also increases with $\theta$, as higher arrival rate incurs longer waiting times, which accumulate additional variability across replications. In contrast, when the arrival rate is relatively smaller than the service rate, most of the randomness arises solely from the arrival and service processes, so the variability is much smaller. 

\begin{figure}[htbp]
    \centering
    \includegraphics[width=0.85\linewidth]{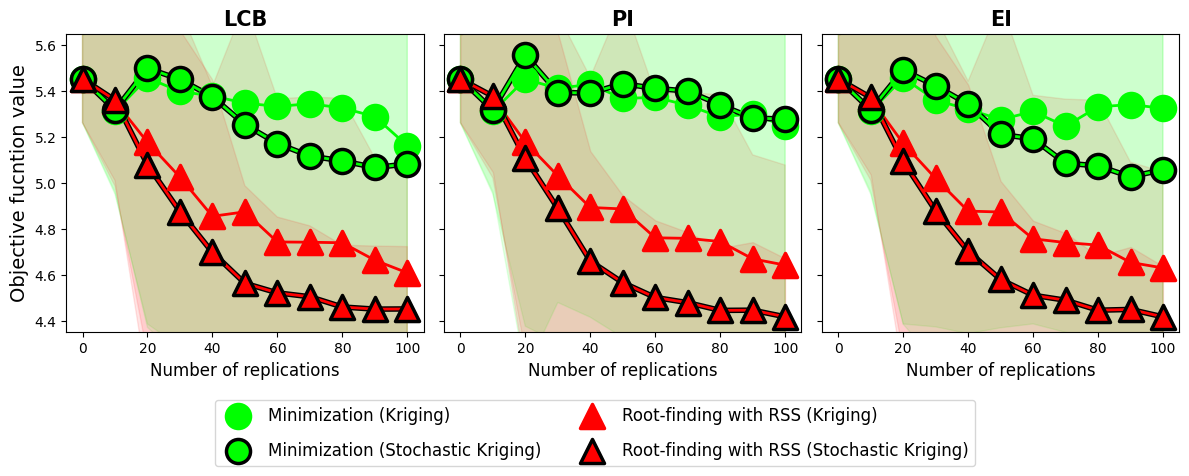}
    \caption{Calibration results of the M/M/1 example. Objective function value at each  best observed solution is obtained from 1,000 post replications and reported with 95\% confidence intervals.}
    \label{siam:fig:mm1_result}
\end{figure}

In this example, the two formulations are expected to yield similar optimal solutions. Since all entities are governed by the same arrival rate $\theta$ and are sequentially dependent, changes in $\theta$ tend to affect the resulting sojourn time of all entities in the same direction. This makes the spatial variability $\delta(\theta)$ small, meaning that the averaged behavior well represents the individual output. 

\Cref{siam:fig:mm1_result} summarizes the calibration results. The root finding framework consistently achieves faster convergence than minimization across all acquisition functions. Similar to the previous example, given the inherent variability of the problem, the benefit of stochastic kriging is not well exploited in the minimization, as seen in LCB and PI cases. Again, this is mostly attributed to the exploratory nature of minimization in the early stage, requiring more observations to identify a promising region. Moreover, the best observed solution under minimization can be misleading in the early stage, which leads to higher objective function value. This is because typically, minimization has to solve more complex functional structure, and with few and often uninformative early observations, the metamodel is unlikely to be reliable.

\subsection{Epidemic Simulation Example}

In this experiment, we consider an epidemic computer model based on a stochastic SIR (Susceptible–Infected–Recovered) framework. In this example, a population of 100 individuals is divided into three compartments: susceptible (S), infected (I), and recovered (R). Initially, 10 individuals are infected, and disease transmission occurs through contacts between infected and susceptible individuals. Each infected individual interacts with up to 2 randomly selected susceptible individuals per day, and each contact results in a new infection with probability $\theta^{\text{real}} = 0.65$. Infected individuals recover independently with a daily probability of 0.7. The simulation runs for 5 days epidemic period, after which the proportion of recovered individuals is recorded as the model output. 
\Cref{siam:fig:epi_example} depicts the example temporal evolutions of this model.

The objective of this experiment is to estimate the unknown infection rate $\theta \in [0,1]$ by comparing the simulated daily recovery proportions ($m_y = 5$) against the observed trajectory. This problem setting is particularly relevant in practice, since during the epidemic outbreaks, recovery records are often the only observable data \cite{kiamari2020covid}.

\Cref{siam:fig:epi_func} illustrates the estimated discrepancy functions. As the infection rate $\theta$ increases, the disease spreads more rapidly in the early stages, leading to a higher proportion of (observed) recovered individuals, thereby the signed discrepancy becomes more negative. Notice that the variability also increases with $\theta$, since a higher infection rate leads to more infected individuals, generating more interactions and further amplifying the stochasticity in the recovery phase. However, compared to the M/M/1 example, the overall variability remains relatively small across the parameter space. In the M/M/1 system, entities are sequentially dependent, so that the effect of $\theta$ tends to accumulate throughout the queue. In contrast, in the SIR model, the same infection and recovery rules apply to all individuals. Although interactions within each replication are processed sequentially, the realized order is itself random, so the individual-level randomness tends to average out at the population level. Hence, $\text{Var}(S(\theta))$ is also small in the SIR problem.

\begin{figure}[!htbp]
    \centering
    \includegraphics[width=0.85\linewidth]{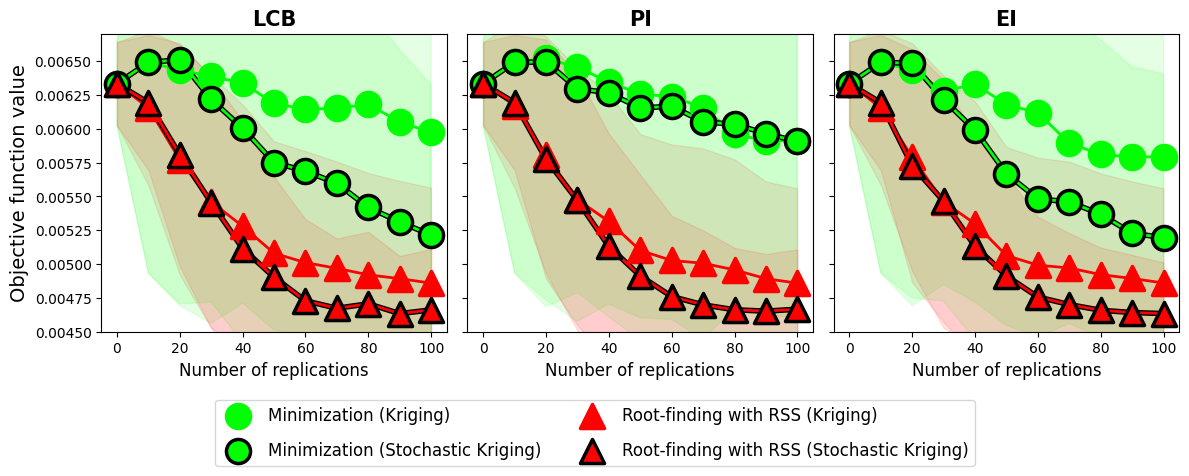}
    \caption{Calibration results of the SIR example. Objective function value at each  best observed solution is obtained from 1,000 post replications and reported with 95\% confidence intervals.}
    \label{siam:fig:epi_result}
\end{figure}

\Cref{siam:fig:epi_result} summarizes the overall experimental results. The root finding framework consistently achieves faster convergence than minimization across all acquisition functions. Compared to the M/M/1 example, in the root-finding framework, due to the relatively small variability of the problem, there is no noticeable performance difference between kriging and stochastic kriging at the early stage. Among the acquisition functions, PI tends to exhibit slower convergence under minimization in both M/M/1 and SIR examples. Overall, root finding framework under stochastic kriging with RSS achieves the best calibration performance across all three examples.

\section{Conclusion}\label{siam:sec:conclusion}

In this paper, we introduced a root finding framework for computer model calibration that leverages sign information to accelerate the calibration process. We proposed three new acquisition functions and analyzed their asymptotic behavior in rootless scenarios. Moreover, a search space reduction technique is introduced that is uniquely enabled by the sign structure of the root finding formulation. Both the theoretical analyses and numerical experiments demonstrate that the proposed framework offers a reliable and computationally efficient alternative to standard calibration approaches.

We close with several remarks. First, while the root finding approach was introduced within a metamodel-based setting, the theoretical results on solution accuracy and convergence in \Cref{siam:sec:sol-acc-rf} are broadly applicable to other derivative-free optimization methods. Second, adaptive sampling presents a promising direction for further improving the computational efficiency, including not only the sequential design of new configurations but also the adaptive refinement of replication budget in the stochastic RSS. Lastly, incorporating variance reduction techniques such as common random numbers \cite{ha2025complexity} or stratified sampling \cite{park2025strata} are also noteworthy.

\appendix

\section{Proof of Solution Accuracy Bound}\label{siam:appendix:sol-acc-proof}

This section provides the proof of \Cref{siam:thm:sol-acc}.
\begin{proof}
Let $\theta_n^* \in \Theta_n^*$ satisfy 
$\mathrm{dist}(\bar{\theta}^*, \Theta_n^*)^2 = \|\theta_n^* - \bar{\theta}^*\|^2$, 
and let $\theta^* \in \Theta^*$ satisfy $\mathrm{dist}(\theta_n^*, \Theta^*) = \|\theta_n^* - \theta^*\|$.
Then, under \Cref{siam:cor:strong-convex},
\begin{align}\label{siam:eq:solution-accuracy}
    \mathrm{dist}(\bar{\theta}^*, \Theta_n^*)^2
    &\leq 
    \frac{2}{\bar{\ell}_C}\left(\bar{f}(\theta_n^*) - \bar{f}(\bar{\theta}^*)\right)
    \leq
    \frac{2}{\bar{\ell}_C}\left(\bar{f}(\theta_n^*) - f(\theta^*)\right)
    \\
    &=
    \frac{2}{\bar{\ell}_C}\left(
        f(\theta_n^*) - f(\theta^*) + \delta(\theta_n^*) + \mathrm{Var}(S(\theta_n^*))
    \right)
    \nonumber
    \\
    &=
    \frac{2}{\bar{\ell}_C}
    \left(
        \underbrace{(\mathbb{E}[S(\theta_n^{*})])^2 - (\mathbb{E}[S(\theta^*)])^2}_{(\text{A}_{1})}
        + \underbrace{\delta(\theta_n^*)}_{(\text{A}_{2})}
        + \underbrace{\mathrm{Var}(S(\theta_n^*))}_{(\text{A}_{3})}
    \right)
    \nonumber,
\end{align}
where the second inequality holds since $f(\theta) \leq \bar{f}(\theta)$ and $f(\theta^*) \leq f(\bar{\theta}^*)$, which gives $f(\theta^*) \leq f(\bar{\theta}^*) \leq \bar{f}(\bar{\theta}^*)$. For $(\text{A}_{1})$, under \Cref{siam:assump:lips} and using $a^2 - b^2 = (a-b)^2 + 2b(a-b)$,
    \begin{align*}
        (\mathbb{E}[S(\theta_n^{*})])^2 - (\mathbb{E}[S(\theta^*)])^2
        &=
        \left(\mathbb{E}[S(\theta_n^{*})] - \mathbb{E}[S(\theta^*)]\right)^2
        +
        2\mathbb{E}[S(\theta^*)]
        \left(\mathbb{E}[S(\theta_n^{*})] - \mathbb{E}[S(\theta^*)]\right)
        \\
        &\leq
        \tilde{\ell}_E^2\mathrm{dist}(\theta_n^*, \Theta^*)^2
        + 2\tilde{\ell}_E\left|\mathbb{E}[S(\theta^*)]\right|\mathrm{dist}(\theta_n^*, \Theta^*).
    \end{align*}
    Since $\theta^* \in \Theta^*$ and under \Cref{siam:assump:rootexist-RF} we have 
    $\Theta^* = \tilde{\Theta}^*$, it follows that $\mathbb{E}[S(\theta^*)]=0$, and thus
    \begin{equation*}
        (\mathbb{E}[S(\theta_n^{*})])^2 
        \leq
        \tilde{\ell}_E^2\mathrm{dist}(\theta_n^*, \Theta^*)^2.
    \end{equation*}
    For $(\text{A}_{2})$, under \Cref{siam:assump:out-var-lip},
    \begin{equation*}
        \delta(\theta_n^{*}) \leq \delta(\theta^*) + \tilde{\ell}_\delta\mathrm{dist}(\theta_n^*, \Theta^*).
    \end{equation*}
    For $(\text{A}_{3})$, applying the law of total variance on $\theta_n^{*}$,
    \begin{equation}\label{siam:eq:var-kr-hat}
        \mathrm{Var}(S(\theta_n^{*}))
        = 
        \underbrace{\mathbb{E}\left[\mathrm{Var}(S(\theta_n^{*}) \mid \theta_n^*)\right]}_{(\text{B}_1)}
        +
        \underbrace{\mathrm{Var}\left(\mathbb{E}[S(\theta_n^{*}) \mid \theta_n^*]\right)}_{(\text{B}_2)}.
    \end{equation}
    We bound the two components separately. For $(\text{B}_1)$, under \Cref{siam:assump:var-lips},
    \begin{equation*}
        -\tilde{\ell}_V\mathrm{dist}(\theta_n^*, \Theta^*)
        \leq
        \mathrm{Var}(S(\theta_n^{*}) \mid \theta_n^*) - \mathrm{Var}(S(\theta^*))
        \leq
        \tilde{\ell}_V\mathrm{dist}(\theta_n^*, \Theta^*).
    \end{equation*}
    Using the right-hand-side inequality and taking expectation gives
    \begin{align*}
        \mathbb{E}\left[\mathrm{Var}(S(\theta_n^{*}) \mid \theta_n^*)\right]
        \leq
        \mathbb{E}\left[\mathrm{Var}(S(\theta^*)) + \tilde{\ell}_V\mathrm{dist}(\theta_n^*, \Theta^*)\right].
    \end{align*}
    Since $\theta^* \in \Theta^*$, this further implies
    \begin{equation*}
        \mathbb{E}\left[\mathrm{Var}(S(\theta_n^{*}) \mid \theta_n^*)\right]
        \leq
        \sup_{\theta \in \Theta^*}\mathrm{Var}(S(\theta))
        + \tilde{\ell}_V\mathbb{E}\left[\mathrm{dist}(\theta_n^*, \Theta^*)\right].
    \end{equation*}
    For $(\text{B}_2)$, from $\mathrm{Var}(X) \leq \mathbb{E}[X^2]$ and under \Cref{siam:assump:lips},
    \begin{align*}
        \mathrm{Var}\left(\mathbb{E}[S(\theta_n^{*}) \mid \theta_n^*]\right)
        &\leq
        \mathbb{E}\left[\left(
            \mathbb{E}[S(\theta_n^{*}) \mid \theta_n^*]
        \right)^2\right]
        \leq
        \tilde{\ell}_E^2\mathbb{E}\left[\mathrm{dist}(\theta_n^*, \Theta^*)^2\right].
    \end{align*}
    Hence, \eqref{siam:eq:var-kr-hat} is bounded as
    \begin{equation*}\label{siam:eq:variance-bound}
        \mathrm{Var}(S(\theta_n^{*}))
        \leq
        \sup_{\theta \in \Theta^*}\mathrm{Var}(S(\theta))
        + \tilde{\ell}_V\mathbb{E}\left[\mathrm{dist}(\theta_n^*, \Theta^*)\right]
        + \tilde{\ell}_E^2\mathbb{E}\left[\mathrm{dist}(\theta_n^*, \Theta^*)^2\right].
    \end{equation*}
    Accordingly, the solution bound in \eqref{siam:eq:solution-accuracy} can be written as
    \begin{align*}
        \mathrm{dist}(\bar{\theta}^*, \Theta_n^*)^2
        &\leq
        \frac{2}{\bar{\ell}_C}
        \Bigg(
            \tilde{\ell}_E^2\mathrm{dist}(\theta_n^*, \Theta^*)^2
            +
            \delta(\theta^*)
            +
            \tilde{\ell}_\delta\mathrm{dist}(\theta_n^*, \Theta^*)
        \\
        &\quad
            + \sup_{\theta \in \Theta^*}\mathrm{Var}(S(\theta))
            + \tilde{\ell}_V\mathbb{E}\left[\mathrm{dist}(\theta_n^*, \Theta^*)\right]
            + \tilde{\ell}_E^2\mathbb{E}\left[\mathrm{dist}(\theta_n^*, \Theta^*)^2\right]
        \Bigg).
    \end{align*}
    Since $\theta^* \in \Theta^*$, it follows that
    \begin{equation*}
        \delta(\theta^*) + \sup_{\theta \in \Theta^*}\mathrm{Var}(S(\theta))
        \leq
        \sup_{\theta \in \Theta^*}\left(\delta(\theta) + \mathrm{Var}(S(\theta))\right).
    \end{equation*}
    Thereby, we can conclude that
    \begin{align*}
        \mathrm{dist}(\bar{\theta}^*, \Theta_n^*)^2
        &\leq
        \frac{2}{\bar{\ell}_C}
        \Bigg(
            \tilde{\ell}_E^2\mathrm{dist}(\theta_n^*, \Theta^*)^2
            +
            \sup_{\theta \in \Theta^*}\delta(\theta)
            +
            \tilde{\ell}_\delta\mathrm{dist}(\theta_n^*, \Theta^*)
        \\
        &\quad
            + \sup_{\theta \in \Theta^*}\mathrm{Var}(S(\theta))
            + \tilde{\ell}_V\mathbb{E}\left[\mathrm{dist}(\theta_n^*, \Theta^*)\right]
            + \tilde{\ell}_E^2\mathbb{E}\left[\mathrm{dist}(\theta_n^*, \Theta^*)^2\right]
        \Bigg).
    \end{align*}
\end{proof}

\section{Proof of Expected Solution Accuracy Bound}\label{siam:appendix:exp-sol-acc-proof}
This section provides the proof of \Cref{siam:cor:exp-sol-acc}.
\begin{proof}
    Let $\theta_n^* \in \Theta_n^*$ satisfy $\mathrm{dist}(\bar{\theta}^*, \Theta_n^*)^2 = \|\theta_n^* - \bar{\theta}^*\|^2$, and let $\theta^* \in \Theta^*$ satisfy $\mathrm{dist}(\theta_n^*, \Theta^*) = \|\theta_n^* - \theta^*\|$. 
    From \Cref{siam:thm:sol-acc}, taking expectation gives
    \begin{align*}
        \mathbb{E}\left[\mathrm{dist}(\bar{\theta}^*, \Theta_n^*)^2\right]
        &\leq
        \frac{2}{\bar{\ell}_C}
        \Bigg(
            2\tilde{\ell}_E^2\mathbb{E}\left[\mathrm{dist}(\theta_n^*, \Theta^*)^2\right]
            +
            \sup_{\theta \in \Theta^*}\left(\delta(\theta) + \mathrm{Var}(S(\theta))\right)
        \\
        &\quad
        + (\tilde{\ell}_\delta + \tilde{\ell}_V)\mathbb{E}\left[\mathrm{dist}(\theta_n^*, \Theta^*)\right]
        \Bigg)
        \\
        &\leq
        \frac{2}{\bar{\ell}_C}
        \Bigg(
            2\tilde{\ell}_E^2\mathbb{E}\left[\mathrm{dist}(\theta_n^*, \Theta^*)^2\right]
            +
            \sup_{\theta \in \Theta^*}\left(\delta(\theta) + \mathrm{Var}(S(\theta))\right)
        \\
        &\quad
        +
        (\tilde{\ell}_\delta + \tilde{\ell}_V)\mathbb{E}\left[\mathrm{dist}(\theta_n^*, \Theta^*)^2\right]^{1/2}
        \Bigg),
    \end{align*}
    where the second inequality follows from the Cauchy--Schwarz inequality $\mathbb{E}\left[\mathrm{dist}(\theta_n^*, \Theta^*)\right] \leq \mathbb{E}\left[\mathrm{dist}(\theta_n^*, \Theta^*)^2\right]^{1/2}$.
\end{proof}

\section{Proof of Tighter Expected Accuracy Bound of Stochastic Approximation Solution}\label{siam:appendix:tight-bound-proof}
This section provides the proof of \Cref{siam:cor:tight-bound}. 
\begin{proof}
For any $\check{\theta}_n^* \in \check{\Theta}_n^*$, from $a - b \leq |b-a|$,
\begin{equation*}
    f(\check{\theta}_n^*) - \check{f}_n(\check{\theta}_n^*) \leq \left|\check{f}_n(\check{\theta}_n^*) - f(\check{\theta}_n^*)\right| \leq \sup_{\theta \in \Theta}\left|\check{f}_n(\theta) - f(\theta)\right|,
\end{equation*}
where we get
\begin{equation}\label{siam:eq:temp1}
    f(\check{\theta}_n^*) 
    \leq \check{f}_n(\check{\theta}_n^*) + \sup_{\theta \in \Theta}\left|\check{f}_n(\theta) - f(\theta)\right|.
\end{equation}
Since $\check{\theta}_n^*$ minimizes $\check{f}_n$ over $\Theta$, $\check{f}_n(\check{\theta}_n^*) \leq \check{f}_n(\theta^*)$, so we can further bound
\begin{equation*}
    f(\check{\theta}_n^*) \leq \check{f}_n(\theta^*) + \sup_{\theta \in \Theta}\left|\check{f}_n(\theta) - f(\theta)\right|.
\end{equation*}
Similarly, for $\theta^* \in \Theta^*$, now, from $b-a \leq |b-a|$,
\begin{equation*}
    \check{f}_n(\theta^*) - f(\theta^*) \leq \left|\check{f}_n(\theta^*) - f(\theta^*)\right| \leq \sup_{\theta \in \Theta}\left|\check{f}_n(\theta) - f(\theta)\right|,
\end{equation*}
which gives 
\begin{equation}\label{siam:eq:temp2}
    \check{f}_n(\theta^*) 
    \leq 
    f(\theta^*) 
    +
    \sup_{\theta \in \Theta}\left|\check{f}_n(\theta) - f(\theta)\right|.
\end{equation}
From \eqref{siam:eq:temp1} and \eqref{siam:eq:temp2}, we have
\begin{equation*}
    f(\check{\theta}_n^*) 
    \leq 
    f(\theta^*) 
    + 2\sup_{\theta \in \Theta}\left|\check{f}_n(\theta) - f(\theta)\right|.
\end{equation*}
Under \Cref{siam:assump:rootexist-RF}, $f(\theta^*) = 0$, so
\begin{equation*}
    f(\check{\theta}_n^*) \leq 2\sup_{\theta \in \Theta}\left|\check{f}_n(\theta) - f(\theta)\right|.
\end{equation*}
Under \Cref{siam:assump:root-stab}, for any $\theta \in \Theta$,
\begin{equation*}
    \mathrm{dist}(\theta, \Theta^*)^2 \leq \frac{1}{\tilde{\ell}_R^2}f(\theta).
\end{equation*}
For $\check{\theta}_n^*$, taking expectations gives
\begin{equation*}
    \mathbb{E}\left[\mathrm{dist}(\check{\theta}_n^*, \Theta^*)^2\right]
    \leq
    \frac{1}{\tilde{\ell}_R^2}
    \mathbb{E}\left[f(\check{\theta}_n^*)\right]
    \leq
    \frac{2}{\tilde{\ell}_R^2}
    \mathbb{E}\left[\sup_{\theta \in \Theta}\left|\check{f}_n(\theta) - f(\theta)\right|\right].
\end{equation*}
By the same argument, for $\theta_n^*$, we get
\begin{equation*}
    \mathbb{E}\left[\mathrm{dist}(\theta_n^*, \Theta^*)^2\right]
    \leq
    \frac{2}{\tilde{\ell}_R^2}
    \mathbb{E}\left[\sup_{\theta \in \Theta}\left|f_n(\theta) - f(\theta)\right|\right].
\end{equation*}
By the Cauchy--Schwarz inequality,
\begin{equation*}
    \mathbb{E}\left[
        \sup_{\theta \in \Theta}
        \left|
            \check f_n(\theta)-f(\theta)
        \right|
    \right]
    \leq
    \left(
        \mathbb{E}\left[
            \sup_{\theta \in \Theta}
            \left(
                \check f_n(\theta)-f(\theta)
            \right)^2
        \right]
    \right)^{1/2},
\end{equation*}
and
\begin{equation*}
    \mathbb{E}\left[
        \sup_{\theta \in \Theta}
        \left|
            f_n(\theta)-f(\theta)
        \right|
    \right]
    \leq
    \left(
        \mathbb{E}\left[
            \sup_{\theta \in \Theta}
            \left(
                f_n(\theta)-f(\theta)
            \right)^2
        \right]
    \right)^{1/2}.
\end{equation*}
Then, under \Cref{siam:assump:better-mse},
\begin{equation*}
    \left(
        \mathbb{E}\left[
            \sup_{\theta \in \Theta}
            \left(
                \check f_n(\theta)-f(\theta)
            \right)^2
        \right]
    \right)^{1/2}
    \leq
    \left(
        \mathbb{E}\left[
            \sup_{\theta \in \Theta}
            \left(
                f_n(\theta)-f(\theta)
            \right)^2
        \right]
    \right)^{1/2},
\end{equation*}
which implies that the upper bound of $\mathbb{E}\left[\left\|\check{\theta}_n^*-\theta^*\right\|^2\right]$ is smaller than that of $\mathbb{E}\left[\left\|\theta_n^*-\theta^*\right\|^2\right]$. Thereby, we can conclude that 
\begin{equation*}
    \overline{\mathbb{E}}\left[
        \mathrm{dist}(\bar{\theta}^*, \check{\Theta}_n^*)^2
    \right]
    \leq
    \overline{\mathbb{E}}\left[
        \mathrm{dist}(\bar{\theta}^*, \Theta_n^*)^2
    \right].
\end{equation*}
\end{proof}

\section{Derivation of Probability of Improvement}\label{siam:appendix:pi}
The event of improvement at a candidate point $\theta$ is defined as the predicted value being smaller than $v^t_{n}$:  
\begin{equation*}
    \mathcal{I}^t_{n}(\theta) := \left\{\eta^t_{n}(\theta) \leq v^t_{n} \right\}.
\end{equation*}
Using the reparameterization trick, we can express the predictive random variable as  
\begin{equation*}
    \eta^t_{n}(\theta) = \mu^t_{n}(\theta) + \sigma^t_{n}(\theta) Z, \quad Z \sim \mathcal{N}(0, 1).
\end{equation*} 
Then, the PI acquisition function is given by  
\begin{align}
    \text{PI}^t_{n}(\theta) 
        &= \mathbb{P}(\mathcal{I}^t_{n}(\theta)) = \mathbb{P}( \mu^t_{n}(\theta) + \sigma^t_{n}(\theta) Z \leq v^t_{n} )  = \Phi\left( \frac{v^t_{n} - \mu^t_{n}(\theta)}{\sigma^t_{n}(\theta)} \right), \nonumber
\end{align}  
where $\Phi(\cdot)$ denotes the standard normal cumulative distribution function (CDF).

\section{Derivation of Expected Improvement}\label{siam:appendix:ei}
The improvement at a candidate point $\theta$ is defined as 
\begin{equation*}
    \text{EI}^t_{n}(\theta) = \mathbb{E}\left[ \max\left(0, v^t_{n} - \eta^t_{n}(\theta) \right) \right].
\end{equation*}
Using the reparameterization trick, we can write 
\begin{equation*}
    \eta^t_{n}(\theta) = \mu^t_{n}(\theta) + \sigma^t_{n}(\theta) Z, \quad Z \sim \mathcal{N}(0, 1).
\end{equation*}
Substituting into the EI expression gives  
\begin{equation*}
    \text{EI}^t_{n}(\theta) = \mathbb{E}\left[ \max\left(0, v^t_{n} - \mu^t_{n}(\theta) - \sigma^t_{n}(\theta) Z \right) \right].
\end{equation*}
Then, applying $\int_{-\infty}^{a} \phi(b)db = \Phi(a)$ and $\int_{-\infty}^{a} b \phi(b) db = -\phi(a)$, 
\begin{align*}
    \text{EI}^t_{n}(\theta) 
    &= \int_{-\infty}^{z^t_{n}(\theta)} \left( v^t_{n} - \mu^t_{n}(\theta) - \sigma^t_{n}(\theta) z \right) \phi(z)dz 
    = (v^t_{n} - \mu^t_{n}(\theta)) \Phi(z^t_{n}(\theta)) 
       + \sigma^t_{n}(\theta) \phi(z^t_{n}(\theta)), \nonumber
\end{align*}
where $\phi(\cdot)$ denotes the probability density function (PDF) of the standard normal distribution.

\section{Derivation of Root-Finding Probability of Improvement}\label{siam:appendix:pirf}

The improvement event at a candidate point $\theta$ is defined as the predicted value being closer to zero than the current best,
\begin{equation*}
    {\widetilde{\mathcal{I}}}^t_{n}(\theta) 
    := 
    \left\{ |\tilde{\eta}^t_{n}(\theta)| \leq |\tilde{v}^t_{n}| \right\},
\end{equation*}
which accounts for improvement in both sign directions symmetrically. Using the reparameterization trick,
\begin{equation*}
    \tilde{\eta}^t_{n}(\theta) 
    = 
    \tilde{\mu}^t_{n}(\theta) + \tilde{\sigma}^t_{n}(\theta) Z, \quad Z \sim \mathcal{N}(0,1),
\end{equation*}
and we can express the closed-form of root-finding PI as:
\begin{align*}
    {\widetilde{\text{PI}}}^t_{n}(\theta)
    &= \mathbb{P}\left({\widetilde{\mathcal{I}}}^t_{n}(\theta) \right) \\
    &= \mathbb{P}\left( -|\tilde{v}^t_{n}| \leq \tilde{\mu}^t_{n}(\theta) 
       + \tilde{\sigma}^t_{n}(\theta) Z \leq |\tilde{v}^t_{n}| \right) \\
    &= \mathbb{P}\left( 
        \frac{-|\tilde{v}^t_{n}| - \tilde{\mu}^t_{n}(\theta)}{\tilde{\sigma}^t_{n}(\theta)}
        \leq Z \leq
        \frac{|\tilde{v}^t_{n}| - \tilde{\mu}^t_{n}(\theta)}{\tilde{\sigma}^t_{n}(\theta)}
       \right) \\
    &= \Phi\left( \frac{|\tilde{v}^t_{n}| - \tilde{\mu}^t_{n}(\theta)}{\tilde{\sigma}^t_{n}(\theta)} \right) 
     - \Phi\left( \frac{-|\tilde{v}^t_{n}| - \tilde{\mu}^t_{n}(\theta)}{\tilde{\sigma}^t_{n}(\theta)} \right).
\end{align*}

\section{Derivation of Root-Finding Expected Improvement}\label{siam:appendix:eirf}

The improvement at a candidate point $\theta$ is defined as
\begin{equation*}
    {\widetilde{\text{EI}}}^t_{n}(\theta) 
    = 
    \mathbb{E}\left[\max(0, |\tilde{v}^t_{n}| - |\tilde{\eta}^t_{n}(\theta)|)\right].
\end{equation*} 
Then, the full expectation form is expressed as 
\begin{equation}\label{siam:eq:eirf-full}
    {\widetilde{\text{EI}}}^t_{n}(\theta) = \int_{-\infty}^{\infty} \max\left(0, |\tilde{v}^t_{n}| - |\tilde{\mu}^t_{n}(\theta) + \tilde{\sigma}^t_{n}(\theta) z| \right) \phi(z) dz.
\end{equation}
Here, observe that the integrand is nonzero only when $|\tilde{\mu}^t_{n}(\theta) + \tilde{\sigma}^t_{n}(\theta) z| \leq |\tilde{v}^t_{n}|$. 
Then, we can write \eqref{siam:eq:eirf-full} as
\begin{equation*}
    {\widetilde{\text{EI}}}^t_{n}(\theta) = \int_{{\ubar{z}}^t_{n}(\theta)}^{{\bar{z}}^t_{n}(\theta)} \left( |\tilde{v}^t_{n}| - |\tilde{\mu}^t_{n}(\theta) + \tilde{\sigma}^t_{n}(\theta) z| \right) \phi(z) dz.
\end{equation*}
Since the absolute value term changes its sign at $\tilde{z}^t_{n}(\theta)$, we divide the integral parts,
\begin{align*}
    {\widetilde{\text{EI}}}^t_{n}(\theta) 
    &= 
    \int_{{\ubar{z}}^t_{n}(\theta)}^{\tilde{z}^t_{n}(\theta)} \left( |\tilde{v}^t_{n}| + \tilde{\mu}^t_{n}(\theta) + \tilde{\sigma}^t_{n}(\theta) z \right) \phi(z) dz
    + 
    \int_{\tilde{z}^t_{n}(\theta)}^{{\bar{z}}^t_{n}(\theta)} \left( |\tilde{v}^t_{n}| - \tilde{\mu}^t_{n}(\theta) - \tilde{\sigma}^t_{n}(\theta) z \right) \phi(z) dz. 
    \nonumber 
\end{align*}
Separating each term gives
\begin{align*}
    {\widetilde{\text{EI}}}^t_{n}(\theta)
    &= |\tilde{v}^t_{n}| \left( \int_{{\ubar{z}}^t_{n}(\theta)}^{\tilde{z}^t_{n}(\theta)} \phi(z) dz + \int_{\tilde{z}^t_{n}(\theta)}^{{\bar{z}}^t_{n}(\theta)} \phi(z) dz \right) 
    \\
    &\quad 
    + 
    \tilde{\mu}^t_{n}(\theta) \left( \int_{{\ubar{z}}^t_{n}(\theta)}^{\tilde{z}^t_{n}(\theta)} \phi(z) dz - \int_{\tilde{z}^t_{n}(\theta)}^{{\bar{z}}^t_{n}(\theta)} \phi(z)dz \right) \nonumber \\
    &\quad + \tilde{\sigma}^t_{n}(\theta) \left( \int_{{\ubar{z}}^t_{n}(\theta)}^{\tilde{z}^t_{n}(\theta)} z \phi(z) dz - \int_{\tilde{z}^t_{n}(\theta)}^{{\bar{z}}^t_{n}(\theta)} z \phi(z) dz \right). \nonumber
\end{align*}
Using $\int_a^b \phi(z) dz = \Phi(b) - \Phi(a), \int_a^b z \phi(z) dz = \phi(a) - \phi(b)$, we get 
\begin{align*}
    {\widetilde{\text{EI}}}^t_{n}(\theta) 
    &= |\tilde{v}^t_{n}| \left( \Phi({\bar{z}}^t_{n}(\theta)) - \Phi({\ubar{z}}^t_{n}(\theta)) \right) 
    \\
    &\quad+ 
    \tilde{\mu}^t_{n}(\theta) \left( 2\Phi(\tilde{z}^t_{n}(\theta)) - \Phi({\bar{z}}^t_{n}(\theta)) - \Phi({\ubar{z}}^t_{n}(\theta)) \right) \nonumber \\
    &\quad - \tilde{\sigma}^t_{n}(\theta) \left( 2\phi(\tilde{z}^t_{n}(\theta)) - \phi({\bar{z}}^t_{n}(\theta)) - \phi({\ubar{z}}^t_{n}(\theta)) \right). \nonumber
\end{align*}

\section{Gradient Derivations for Root finding Acquisition Functions}\label{siam:sec:appendix:gradient derivation}

In this appendix, we provide analytical expressions for the first-order derivatives of the proposed root finding acquisition functions. We first present the gradients of the posterior mean and variance under both kriging and stochastic kriging models. The gradient of the posterior mean under the kriging model is given by:
\begin{align*}
    \frac{\partial \tilde{\mu}^t_{n}(\theta)}{\partial \theta}
    = 
    \frac{\partial}{\partial \theta}
    \left(
        \mathbf{k}(\theta, \Theta^t; l)
        [\mathbf{k}(\Theta^t, \Theta^t; l)]^{-1}
        \tilde{\boldsymbol{f}}^t_{n}
    \right) 
    = 
    \frac{\partial \mathbf{k}(\theta, \Theta^t; l)}{\partial \theta}
    [\mathbf{k}(\Theta^t, \Theta^t; l)]^{-1}
    \tilde{\boldsymbol{f}}^t_{n}. \nonumber
\end{align*}
Similarly, the gradient of the posterior mean under the stochastic kriging model is:
\begin{align*}
    \frac{\partial \tilde{\mu}^t_{n}(\theta)}{\partial \theta}
    &= \frac{\partial \mathbf{k}(\theta, \Theta^t; l)}{\partial \theta}
    [ \mathbf{k}(\Theta^t, \Theta^t; l) + \Sigma^{(t)} ]^{-1}
    \tilde{\boldsymbol{f}}^t_{n}.
\end{align*}
The gradient of the posterior standard deviation under the kriging model is:
\begin{align*}
    \frac{\partial \tilde{\sigma}^t_{n}(\theta)}{\partial \theta}
    &= 
    \frac{\partial}{\partial \theta} \sqrt{(\tilde{\sigma}^t_{n})^2(\theta)} 
    = 
    \frac{1}{2 \sqrt{(\tilde{\sigma}^t_{n})^2(\theta)}} \cdot \frac{\partial (\tilde{\sigma}^t_{n})^2(\theta)}{\partial \theta} \nonumber \\
    &= -\frac{1}{2 \sqrt{(\tilde{\sigma}^t_{n})^2(\theta)}} \Bigg(
    \frac{\partial \mathbf{k}(\theta, \Theta^t; l)}{\partial \theta} [\mathbf{k}(\Theta^t, \Theta^t; l)]^{-1} \mathbf{k}(\theta, \Theta^t; l)^\top \nonumber 
    \\
    &\quad 
    + 
    \mathbf{k}(\theta, \Theta^t; l) [\mathbf{k}(\Theta^t, \Theta^t; l)]^{-1} \frac{\partial \mathbf{k}(\theta, \Theta^t; l)^\top}{\partial \theta}
    \Bigg), \nonumber
\end{align*}
where
\begin{equation*}
    \frac{\partial \mathbf{k}(\theta, \Theta^t; l)^\top}{\partial \theta}
    =
    \begin{bmatrix}
        -\frac{1}{l^2} (\theta - \theta_1) k(\theta, \theta_1; l) \\
        -\frac{1}{l^2} (\theta - \theta_2) k(\theta, \theta_2; l) \\
        \vdots \\
        -\frac{1}{l^2} (\theta - \theta_{p+t}) k(\theta, \theta_{p+t}; l)
    \end{bmatrix}.
\end{equation*}
The gradient of the posterior standard deviation under the stochastic kriging model is:
\begin{align*}
    \frac{\partial \tilde{\sigma}^t_{n}(\theta)}{\partial \theta}
    &= -\frac{1}{2 \sqrt{(\tilde{\sigma}^t_{n})^2(\theta)}} \Bigg(
    \frac{\partial \mathbf{k}(\theta, \Theta^t; l)}{\partial \theta} [\mathbf{k}(\Theta^t, \Theta^t; l) + \Sigma^{(t)}]^{-1} \mathbf{k}(\theta, \Theta^t; l)^\top 
    \\
    &\quad 
    + \mathbf{k}(\theta, \Theta^t; l) [\mathbf{k}(\Theta^t, \Theta^t; l) + \Sigma^{(t)}]^{-1} \frac{\partial \mathbf{k}(\theta, \Theta^t; l)^\top}{\partial \theta}
    \Bigg)
    \nonumber
    .
\end{align*}
\subsection*{Gradient of $\widetilde{\mathrm{LCB}}$}
\begin{align*}
    \frac{\partial {\widetilde{\text{LCB}}}^t_{n}(\theta)}{\partial \theta}
    &= \frac{\partial}{\partial \theta}\big(|\tilde{\mu}^t_{n}(\theta)| - \kappa \tilde{\sigma}^t_{n}(\theta)\big) = \text{sign}(\tilde{\mu}^t_{n}(\theta)) \cdot \frac{\partial \tilde{\mu}^t_{n}(\theta)}{\partial \theta}
    - \kappa \cdot \frac{\partial \tilde{\sigma}^t_{n}(\theta)}{\partial \theta}. \nonumber
\end{align*}
\subsection*{Gradient of $\widetilde{\mathrm{PI}}$}
We use $\Phi(a) = \int_{-\infty}^{a} \phi(b) db$,
\begin{align*}
    \frac{\partial {\widetilde{\text{PI}}}^t_{n}(\theta)}{\partial \theta}
    &= 
    \frac{\partial}{\partial \theta} \big( \Phi({\bar{z}}^t_{n}(\theta)) - \Phi({\ubar{z}}^t_{n}(\theta)) \big) 
    \\
    &= 
    \frac{\partial}{\partial \theta} \left( \int_{-\infty}^{{\bar{z}}^t_{n}(\theta)} \phi(z) dz - \int_{-\infty}^{{\ubar{z}}^t_{n}(\theta)} \phi(z) dz \right) \nonumber \\
    &= \frac{\partial}{\partial \theta} \int_{{\ubar{z}}^t_{n}(\theta)}^{{\bar{z}}^t_{n}(\theta)} \phi(z) dz. \nonumber
\end{align*}
Applying the Leibniz integral rule $\frac{\partial}{\partial \theta} \int_{a(\theta)}^{b(\theta)} g(z) dz = g(b(\theta)) \frac{\partial b(\theta)}{\partial \theta} - g(a(\theta)) \frac{\partial a(\theta)}{\partial \theta}$, we get
\begin{equation*}
    \frac{\partial {\widetilde{\text{PI}}}^t_{n}(\theta)}{\partial \theta}
    = 
    \phi({\bar{z}}^t_{n}(\theta)) \frac{\partial {\bar{z}}^t_{n}(\theta)}{\partial \theta}
     - \phi({\ubar{z}}^t_{n}(\theta)) \frac{\partial {\ubar{z}}^t_{n}(\theta)}{\partial \theta}
     ,
\end{equation*}
where the gradients of ${\bar{z}}^t_{n}(\theta)$ and ${\ubar{z}}^t_{n}(\theta)$ are
\begin{align*}
    \frac{\partial {\bar{z}}^t_{n}(\theta)}{\partial \theta}
    &= 
    \frac{-\frac{\partial \tilde{\mu}^t_{n}(\theta)}{\partial \theta} \cdot \tilde{\sigma}^t_{n}(\theta)
    - (|\tilde{v}^t_{n}| - \tilde{\mu}^t_{n}(\theta)) \cdot \frac{\partial \tilde{\sigma}^t_{n}(\theta)}{\partial \theta}}{(\tilde{\sigma}^t_{n})^2(\theta)}, 
    \\
    \frac{\partial {\ubar{z}}^t_{n}(\theta)}{\partial \theta}
    &= \frac{-\frac{\partial \tilde{\mu}^t_{n}(\theta)}{\partial \theta} \cdot \tilde{\sigma}^t_{n}(\theta)
    + (|\tilde{v}^t_{n}| + \tilde{\mu}^t_{n}(\theta)) \cdot \frac{\partial \tilde{\sigma}^t_{n}(\theta)}{\partial \theta}}{(\tilde{\sigma}^t_{n})^2(\theta)}
    \nonumber
    .
\end{align*}
\subsection*{Gradient of $\widetilde{\mathrm{EI}}$}
\begin{align*}
    \frac{\partial}{\partial \theta} {\widetilde{\text{EI}}}^t_{n}(\theta) 
    &= |\tilde{v}^t_{n}| \left(
        \frac{\partial {\bar{z}}^t_{n}(\theta)}{\partial \theta} \cdot \phi({\bar{z}}^t_{n}(\theta))
      - \frac{\partial {\ubar{z}}^t_{n}(\theta)}{\partial \theta} \cdot \phi({\ubar{z}}^t_{n}(\theta))
    \right) 
    \\
    &\quad 
    + 
    \frac{\partial \tilde{\mu}^t_{n}(\theta)}{\partial \theta}
        (2\Phi(\tilde{z}^t_{n}(\theta)) - \Phi({\bar{z}}^t_{n}(\theta)) - \Phi({\ubar{z}}^t_{n}(\theta))) \nonumber \\
    &\quad + \tilde{\mu}^t_{n}(\theta) \left(
        2 \frac{\partial \tilde{z}^t_{n}(\theta)}{\partial \theta} \cdot \phi(\tilde{z}^t_{n}(\theta))
      - \frac{\partial {\bar{z}}^t_{n}(\theta)}{\partial \theta} \cdot \phi({\bar{z}}^t_{n}(\theta))
      - \frac{\partial {\ubar{z}}^t_{n}(\theta)}{\partial \theta} \cdot \phi({\ubar{z}}^t_{n}(\theta))
    \right) \nonumber \\
    &\quad - \frac{\partial \tilde{\sigma}^t_{n}(\theta)}{\partial \theta}
        (2\phi(\tilde{z}^t_{n}(\theta)) - \phi({\bar{z}}^t_{n}(\theta)) - \phi({\ubar{z}}^t_{n}(\theta))) \nonumber \\
    &\quad - \tilde{\sigma}^t_{n}(\theta) \Bigg(
        -2 \frac{\partial \tilde{z}^t_{n}(\theta)}{\partial \theta} \cdot \frac{\tilde{z}^t_{n}(\theta)}{\sqrt{2\pi}} e^{-\tfrac{1}{2}\tilde{z}^t_{n}(\theta)^2}
        + \frac{\partial {\ubar{z}}^t_{n}(\theta)}{\partial \theta} \cdot \frac{{\ubar{z}}^t_{n}(\theta)}{\sqrt{2\pi}} e^{-\tfrac{1}{2}{\ubar{z}}^t_{n}(\theta)^2} \nonumber \\
    &\quad + \frac{\partial {\bar{z}}^t_{n}(\theta)}{\partial \theta} \cdot \frac{{\bar{z}}^t_{n}(\theta)}{\sqrt{2\pi}} e^{-\tfrac{1}{2}{\bar{z}}^t_{n}(\theta)^2}
    \Bigg). \nonumber
\end{align*}
The gradient of $\tilde{z}^t_{n}(\theta)$ is given by
\begin{align*}
    \frac{\partial \tilde{z}^t_{n}(\theta)}{\partial \theta}
    = \frac{-\frac{\partial \tilde{\mu}^t_{n}(\theta)}{\partial \theta} \tilde{\sigma}^t_{n}(\theta) + \tilde{\mu}^t_{n}(\theta) \frac{\partial \tilde{\sigma}^t_{n}(\theta)}{\partial \theta}}{(\tilde{\sigma}^t_{n})^2(\theta)}.
\end{align*}


\section{Validation of Analytical Gradients Using Finite Difference Estimates}\label{siam:sec:appendix:gradient validation}

In this experiment, we validate the derived analytical gradients by comparing them against central finite difference estimates with step size $10^{-5}$. 

\begin{figure}[H]
    \centering
    \includegraphics[width=0.8\linewidth]{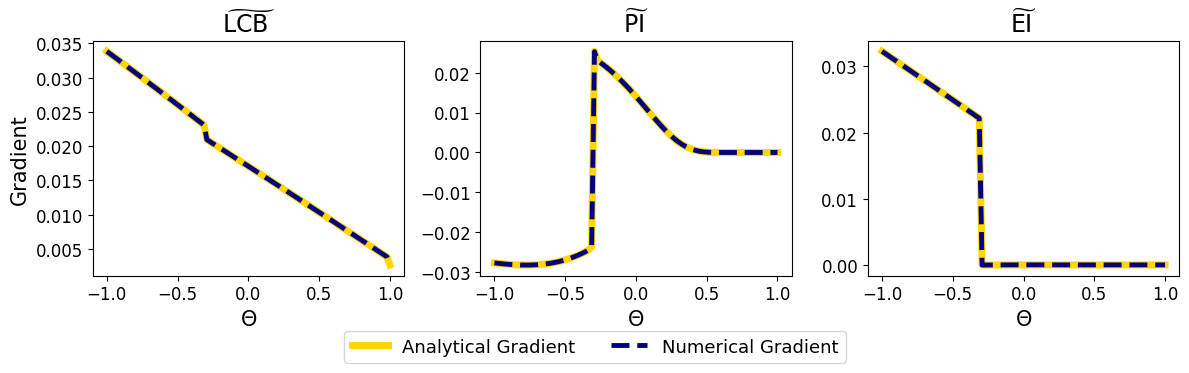}
    \caption{Comparison of analytical and numerical gradient estimates for root finding acquisition functions.}
    \label{siam:fig:gradient-validation}
\end{figure}

As shown in \Cref{siam:fig:gradient-validation}, the analytical gradients are consistent with the finite difference estimates across all three root finding acquisition functions.

\section{Asymptotic Behavior of Root Finding Acquisition Functions}\label{siam:sec:rootless}
In this section, we analyze the rootless scenario, where the sample average root finding objective does not contain a root over $\Theta$. We show that, under suitable posterior regularity conditions, the proposed root finding acquisition functions asymptotically recover the same optimizer as their standard minimization-based counterparts. For analytical convenience, let $(\Omega, \mathcal{F}, \mathbb{P})$ be the probability space in \Cref{siam:sec:rf-metamodeling}, and fix $\omega \in \Omega_+$. The asymptotic analysis throughout this section is with respect to $t$.

\begin{assumption}[Regularity Conditions for the Rootless Scenario]\label{siam:assump:rootless-reg}~
\begin{enumerate}[label=(\alph*), ref=\theassumption(\alph*)]
    \item\label{siam:assump:rootless-one-sided}
    The sample average root finding function is strictly positive over $\Theta$, i.e., there exists $\epsilon(\omega) > 0$ such that $\tilde f_n(\theta;\omega) \geq \epsilon(\omega)$ for all $\theta \in \Theta$.
    \item\label{siam:assump:rootless-globopt}
    Each acquisition function is globally optimized at every iteration.
    \item\label{siam:assump:mean-sign}
    Under the rootless regime, the predictive mean remains strictly positive over $\Theta$ for all sufficiently large $t$, i.e., $\tilde{\mu}^t(\theta;\omega) > 0$ for all $\theta \in \Theta$.
    
    \item\label{siam:assump:rootless-unique}
    For all sufficiently large $t$, the compared acquisition functions have unique optimizers over $\Theta$.
\end{enumerate}
\end{assumption}
\begin{assumption}[Regularity Conditions for the Small-$\epsilon$ Regime]\label{siam:assump:rootless-small}~
\begin{enumerate}[label=(\alph*), ref=\theassumption(\alph*)]
    \item\label{siam:assump:asymptotic-order}
    Let $\Theta_t^* \subseteq \Theta$ be a neighborhood containing the optimizer of the acquisition function. As $t \to \infty$,
    \[
        \sup_{\theta \in \Theta_t^*} \frac{\tilde{\mu}^t(\theta;\omega)}{|\tilde{v}^t(\omega)|} \to 0,
        \qquad
        \sup_{\theta \in \Theta_t^*} \frac{|\tilde{v}^t(\omega)|}{\tilde{\sigma}^t(\theta;\omega)} \to 0.
    \]
    \item\label{siam:assump:local-variance}
    Let $\Theta_t^* \subseteq \Theta$ be a neighborhood containing the optimizer of the acquisition function. The predictive standard deviation is asymptotically constant over $\Theta_t^*$, i.e.,
    \[
        \sup_{\theta_a,\theta_b \in \Theta_t^*}
        \left|\tilde\sigma^t(\theta_a;\omega) - \tilde\sigma^t(\theta_b;\omega)\right|
        \to 0,
        \qquad \text{as } t \to \infty.
    \]
\end{enumerate}
\end{assumption}

\Cref{siam:assump:rootless-one-sided} formalizes the rootless regime, where the sample average root finding function is strictly positive across the parameter space. Although the following analysis is presented under the strictly positive case, the arguments and proofs extend symmetrically to the negative case. \Cref{siam:assump:rootless-globopt} assures that the acquisition function is globally optimized at each iteration. This technical condition is imposed to avoid potential complications arising from local optimization of the acquisition function. \Cref{siam:assump:mean-sign} is a mild condition that the predictive mean follows the one-sided structure of the rootless regime. Lastly, \Cref{siam:assump:rootless-unique} ensures that the compared acquisition functions admit unique optimizers for all sufficiently large $t$. 

The conditions in \Cref{siam:assump:rootless-small} are invoked only in the small-$\epsilon$ analyses of PI and EI, where $\epsilon(\omega):= \inf_{\theta \in \Theta}\tilde{f}_n(\theta;\omega) > 0$ is small enough that the best observed value $|\tilde{v}^t(\omega)|$ and the predictive mean $\tilde{\mu}^t(\theta;\omega)$ both approach zero as $t \to \infty$. \Cref{siam:assump:asymptotic-order} reflects the local asymptotic behavior of the search in this regime, where regions favored by improvement-based acquisition functions are expected to have predictive means smaller than the current best observed value, since otherwise the standardized improvement term would be strictly negative, yielding a small CDF value. The second relation captures the presence of regions with sufficient remaining uncertainty to explore, as the search would have otherwise already converged. \Cref{siam:assump:local-variance} requires the predictive standard deviation to vary only mildly over the neighborhood $\Theta^*_t$ containing the optimizer, and is more likely to be satisfied at a later stage of the search, once $\Theta^*_t$ is well-identified and enough evaluations have accumulated within it so that the predictive uncertainty stabilizes locally.

\subsection*{Root finding Lower Confidence Bound in Rootless Scenario}
\begin{theorem}\label{siam:thm:lcb-rf-rootless}
    Under \Cref{siam:assump:rootless-reg}, for a fixed $\omega \in \Omega_+$, the optimizer of the root finding and standard LCB acquisition functions
    coincide asymptotically, i.e.,
    \[
        \left\|\arg\min_{\theta \in \Theta} \widetilde{\mathrm{LCB}}^t(\theta;\omega)
        - \arg\min_{\theta \in \Theta} \mathrm{LCB}^t(\theta;\omega)\right\| \to 0,
        \quad \text{as } t \to \infty.
    \]
\end{theorem}
\begin{proof}[Proof of \Cref{siam:thm:lcb-rf-rootless}]
    The standard LCB acquisition function and its root finding counterpart at iteration $t$ are defined as
    \[
        \mathrm{LCB}^t(\theta;\omega) = \mu^t(\theta;\omega) - \kappa\sigma^t(\theta;\omega),
        \quad
        \widetilde{\mathrm{LCB}}^t(\theta;\omega) = |\tilde{\mu}^t(\theta;\omega)| - \kappa\tilde{\sigma}^t(\theta;\omega).
    \]
    Under \Cref{siam:assump:mean-sign}, the predictive mean satisfies
    $\tilde{\mu}^t(\theta;\omega) > 0$ for all $\theta \in \Theta$ and all
    sufficiently large $t$, which implies
    \[
        |\tilde{\mu}^t(\theta;\omega)| = \tilde{\mu}^t(\theta;\omega),
        \qquad \forall \theta \in \Theta.
    \]
    Therefore,
    \[
        \widetilde{\mathrm{LCB}}^t(\theta;\omega) = \mathrm{LCB}^t(\theta;\omega),
        \qquad \forall \theta \in \Theta,
    \]
    and under \Cref{siam:assump:rootless-unique}, both acquisition functions have same unique optimizers for all sufficiently large $t$. Then, under \Cref{siam:assump:rootless-globopt},
    \[
        \left\|\arg\min_{\theta \in \Theta} \widetilde{\mathrm{LCB}}^t(\theta;\omega)
        - \arg\min_{\theta \in \Theta} \mathrm{LCB}^t(\theta;\omega)\right\| \to 0,
        \quad \text{as } t \to \infty.
    \]
\end{proof}
Notably, although \Cref{siam:thm:lcb-rf-rootless} is established in the asymptotic regime, this reverting behavior can generally be observed in practice even with a relatively small number of evaluations. Under \Cref{siam:assump:rootless-one-sided}, the observed evaluations remain predominantly one-sided, and the predictive mean $\tilde{\mu}^t(\theta;\omega)$ therefore also tends to be one-sided over $\Theta$.

\subsection*{Root finding Probability Improvement in Rootless Scenario}

\begin{theorem}\label{siam:thm:pi-rf-rootless}
    Under \Cref{siam:assump:rootless-reg} and \Cref{siam:assump:asymptotic-order}, for a fixed $\omega \in \Omega_+$, the maximizers of the root finding and standard PI acquisition functions coincide asymptotically, i.e.,
    \[
        \left\|\arg\max_{\theta \in \Theta} \widetilde{\mathrm{PI}}^t(\theta;\omega)
        - \arg\max_{\theta \in \Theta} \mathrm{PI}^t(\theta;\omega)\right\| \to 0,
        \quad \text{as } t \to \infty.
    \]
\end{theorem}

\begin{proof}[Proof of \Cref{siam:thm:pi-rf-rootless}]
    The standard PI acquisition function and its root finding counterpart at iteration $t$ are defined as
    \begin{align*}
        \mathrm{PI}^t(\theta;\omega)
        &= \Phi\left( \frac{v^t(\omega) - \mu^t(\theta;\omega)}{\sigma^t(\theta;\omega)} \right)
        ,
        \\
        \widetilde{\mathrm{PI}}^t(\theta;\omega)
        &= \Phi\left( \frac{|\tilde{v}^t(\omega)| - \tilde{\mu}^t(\theta;\omega)}{\tilde{\sigma}^t(\theta;\omega)} \right)
        - \Phi\left( \frac{-|\tilde{v}^t(\omega)| - \tilde{\mu}^t(\theta;\omega)}{\tilde{\sigma}^t(\theta;\omega)} \right).
    \end{align*}
    We consider two distinct regimes based on the magnitude of $\epsilon(\omega)$.

    \paragraph{Case 1. Large-$\epsilon$} 
    
    We show that the second term in $\widetilde{\mathrm{PI}}^t(\theta;\omega)$ is uniformly negligible over $\Theta$ as $t \to \infty$. Since $|\tilde{v}^t(\omega)| \geq \epsilon(\omega)$ and $\tilde{\mu}^t(\theta;\omega) \geq \epsilon(\omega)$ for sufficiently large $t$, we have
    \[
        \frac{-|\tilde{v}^t(\omega)| - \tilde{\mu}^t(\theta;\omega)}{\tilde{\sigma}^t(\theta;\omega)}
        \leq -\frac{2\epsilon(\omega)}{\tilde{\sigma}^t(\theta;\omega)}.
    \]
    By Mills' ratio, for $a > 0$, $\Phi(-a) \leq \frac{1}{a\sqrt{2\pi}}e^{-a^2/2}$, so that
    \[
        \Phi\left(\frac{-|\tilde{v}^t(\omega)| - \tilde{\mu}^t(\theta;\omega)}{\tilde{\sigma}^t(\theta;\omega)}\right)
        \leq \frac{1}{\sqrt{2\pi}} \frac{\tilde{\sigma}^t(\theta;\omega)}{2\epsilon(\omega)}
        \exp\left(-\frac{2\epsilon(\omega)^2}{\tilde{\sigma}^t(\theta;\omega)^{2}}\right).
    \]
    Since $\tilde{\sigma}^t(\theta;\omega) \to 0$ as $t \to \infty$, the argument $-2\epsilon(\omega) / \tilde{\sigma}^t(\theta;\omega) \to -\infty$, and hence the second term vanishes uniformly over $\Theta$. Therefore, the root finding acquisition function reduces to
    \[
        \widetilde{\mathrm{PI}}^t(\theta;\omega) \to
        \Phi\left(\frac{|\tilde{v}^t(\omega)| - \tilde{\mu}^t(\theta;\omega)}{\tilde{\sigma}^t(\theta;\omega)}\right)
        = \mathrm{PI}^t(\theta;\omega).
    \]
    Then, under \Cref{siam:assump:rootless-globopt} and \Cref{siam:assump:rootless-unique}, we obtain
    \[
        \left\|\arg\max_{\theta \in \Theta} \widetilde{\mathrm{PI}}^t(\theta;\omega)
        - \arg\max_{\theta \in \Theta} \mathrm{PI}^t(\theta;\omega)\right\| \to 0,
        \quad \text{as } t \to \infty.
    \]

    \paragraph{Case 2. Small-$\epsilon$}

    We now consider the regime where the irreducible discrepancy $\epsilon(\omega) > 0$ is small. Let $\Theta^*_t \subseteq \Theta$ be a compact neighborhood containing the maximizer of the acquisition function. Under \Cref{siam:assump:asymptotic-order}, both arguments of $\Phi$ in $\widetilde{\mathrm{PI}}^t(\theta;\omega)$ approach zero, i.e.,
    \[
        \frac{|\tilde{v}^t(\omega)| - \tilde{\mu}^t(\theta;\omega)}{\tilde{\sigma}^t(\theta;\omega)} \to 0, \qquad
        \frac{-|\tilde{v}^t(\omega)| - \tilde{\mu}^t(\theta;\omega)}{\tilde{\sigma}^t(\theta;\omega)} \to 0, \quad \text{as } t \to \infty.
    \]
    We then apply a second-order Taylor expansion $\Phi(a) = 1/2 + a\phi(0) + O(a^3)$ around $a = 0$, which gives, for $\theta \in \Theta^*_t$,
    \begin{align*}
        \mathrm{PI}^t(\theta;\omega) &= \frac{1}{2} + \phi(0) \cdot \frac{v^t(\omega) - \mu^t(\theta;\omega)}{\sigma^t(\theta;\omega)} + O\left(\left(\frac{v^t(\omega) - \mu^t(\theta;\omega)}{\sigma^t(\theta;\omega)}\right)^3\right), \\
        \widetilde{\mathrm{PI}}^t(\theta;\omega) &= 2\phi(0) \cdot \frac{|\tilde{v}^t(\omega)|}{\tilde{\sigma}^t(\theta;\omega)} + O\left(\left(\frac{|\tilde{v}^t(\omega)| + \tilde{\mu}^t(\theta;\omega)}{\tilde{\sigma}^t(\theta;\omega)}\right)^3\right).
    \end{align*}
    Under \Cref{siam:assump:mean-sign}, for sufficiently large $t$ we have $|\tilde{\mu}^t(\theta;\omega)| = \tilde{\mu}^t(\theta;\omega)$ and $|\tilde{v}^t(\omega)| = \tilde{v}^t(\omega)$, so that
    \[
        \widetilde{\mathrm{PI}}^t(\theta;\omega) = 2 \cdot \mathrm{PI}^t(\theta;\omega) - 1 + \underbrace{2\phi(0) \cdot \frac{\tilde{\mu}^t(\theta;\omega)}{\tilde{\sigma}^t(\theta;\omega)} + O\left(\left(\frac{\tilde{v}^t(\omega) + \tilde{\mu}^t(\theta;\omega)}{\tilde{\sigma}^t(\theta;\omega)}\right)^3\right)}_{\tilde{r}^t(\theta;\omega)}.
    \]
    Under \Cref{siam:assump:asymptotic-order}, $\tilde{r}^t(\theta;\omega) \to 0$ over $\theta \in \Theta^*_t$, and therefore
    \[
        \widetilde{\mathrm{PI}}^t(\theta;\omega) = 2 \cdot \mathrm{PI}^t(\theta;\omega) - 1, \quad \text{as } t \to \infty.
    \]
    Since $\widetilde{\mathrm{PI}}^t(\theta;\omega)$ is a positive linear transformation of $\mathrm{PI}^t(\theta;\omega)$, both share the same maximizer, and under \Cref{siam:assump:rootless-globopt} and \Cref{siam:assump:rootless-unique},
    \[
        \left\|\arg\max_{\theta \in \Theta} \widetilde{\mathrm{PI}}^t(\theta;\omega) - \arg\max_{\theta \in \Theta} \mathrm{PI}^t(\theta;\omega)\right\| \to 0, \quad \text{as } t \to \infty.
    \]
\end{proof}

\subsection*{Root finding Expected Improvement in Rootless Scenario}

\begin{theorem}\label{siam:thm:ei-rf-rootless}
    Under \Cref{siam:assump:rootless-reg} and \Cref{siam:assump:rootless-small}, for a fixed $\omega \in \Omega_+$, the maximizers of the root finding and standard EI acquisition functions coincide asymptotically, i.e.,
    \[
        \left\|\arg\max_{\theta \in \Theta} \widetilde{\mathrm{EI}}^t(\theta;\omega)
        - \arg\max_{\theta \in \Theta} \mathrm{EI}^t(\theta;\omega)\right\| \to 0, \quad \text{as } t \to \infty.
    \]
\end{theorem}

\begin{proof}[Proof of \Cref{siam:thm:ei-rf-rootless}]
                For EI, we follow the same logical flow of \Cref{siam:thm:pi-rf-rootless}, exploring the large-$\epsilon$ and small-$\epsilon$ regimes.

        \paragraph{Case 1. Large-$\epsilon$}

        Recall the closed-form expression of $\widetilde{\mathrm{EI}}^t(\theta;\omega)$ is expressed as
        \begin{align*}
            \widetilde{\mathrm{EI}}^t(\theta;\omega)
            &= \tilde{v}^t(\omega) \left[ \Phi\left( \frac{\tilde{v}^t(\omega) - \tilde{\mu}^t(\theta;\omega)}{\tilde{\sigma}^t(\theta;\omega)} \right)
            - \Phi\left( \frac{-\tilde{v}^t(\omega) - \tilde{\mu}^t(\theta;\omega)}{\tilde{\sigma}^t(\theta;\omega)} \right) \right] \\
            &\quad + \tilde{\mu}^t(\theta;\omega) \left[ 2\Phi\left( \frac{-\tilde{\mu}^t(\theta;\omega)}{\tilde{\sigma}^t(\theta;\omega)} \right)
            - \Phi\left( \frac{\tilde{v}^t(\omega) - \tilde{\mu}^t(\theta;\omega)}{\tilde{\sigma}^t(\theta;\omega)} \right)
            - \Phi\left( \frac{-\tilde{v}^t(\omega) - \tilde{\mu}^t(\theta;\omega)}{\tilde{\sigma}^t(\theta;\omega)} \right) \right] \\
            &\quad - \tilde{\sigma}^t(\theta;\omega) \left[ 2\phi\left( \frac{-\tilde{\mu}^t(\theta;\omega)}{\tilde{\sigma}^t(\theta;\omega)} \right)
            - \phi\left( \frac{\tilde{v}^t(\omega) - \tilde{\mu}^t(\theta;\omega)}{\tilde{\sigma}^t(\theta;\omega)} \right)
            - \phi\left( \frac{-\tilde{v}^t(\omega) - \tilde{\mu}^t(\theta;\omega)}{\tilde{\sigma}^t(\theta;\omega)} \right) \right].
        \end{align*}
        Again using Mills' ratio, we have
        \begin{align*}
            \Phi\left(\frac{\tilde{v}^t(\omega)-\tilde{\mu}^t(\theta;\omega)}{\tilde{\sigma}^t(\theta;\omega)}\right)
            &\leq \frac{1}{\sqrt{2\pi}} \frac{\tilde{\sigma}^t(\theta;\omega)}{\tilde{\mu}^t(\theta;\omega)-\tilde{v}^t(\omega)}
            \exp\left(-\frac{(\tilde{\mu}^t(\theta;\omega)-\tilde{v}^t(\omega))^{2}}{2\tilde{\sigma}^t(\theta;\omega)^{2}}\right), \\
            \Phi\left(\frac{-\tilde{\mu}^t(\theta;\omega)}{\tilde{\sigma}^t(\theta;\omega)}\right)
            &\leq \frac{1}{\sqrt{2\pi}} \frac{\tilde{\sigma}^t(\theta;\omega)}{\tilde{\mu}^t(\theta;\omega)} \exp\left(-\frac{\tilde{\mu}^t(\theta;\omega)^2}{2\tilde{\sigma}^t(\theta;\omega)^{2}}\right).
        \end{align*}
        As $t \to \infty$ and $\tilde{\sigma}^t(\theta;\omega) \to 0$, we can show that
        \[
            \frac{\Phi\left( \frac{\tilde{v}^t(\omega) - \tilde{\mu}^t(\theta;\omega)}{\tilde{\sigma}^t(\theta;\omega)} \right)}{\Phi\left( \frac{-\tilde{\mu}^t(\theta;\omega)}{\tilde{\sigma}^t(\theta;\omega)} \right)} \leq \frac{\epsilon(\omega)}{\tilde{\mu}^t(\theta;\omega) - \tilde{v}^t(\omega)}\exp\left(\frac{\epsilon(\omega)^2 - (\tilde{\mu}^t(\theta;\omega) - \tilde{v}^t(\omega))^2}{2\tilde{\sigma}^t(\theta;\omega)^2}\right) \to \infty,
        \]
        and thus $\Phi\left( \frac{\tilde{v}^t(\omega) - \tilde{\mu}^t(\theta;\omega)}{\tilde{\sigma}^t(\theta;\omega)} \right)$ becomes the dominant term among the CDF terms. For the PDF terms, we get
        \begin{align*}
            \frac{\phi\left(\frac{\tilde{v}^t(\omega)-\tilde{\mu}^t(\theta;\omega)}{\tilde{\sigma}^t(\theta;\omega)}\right)}{\phi\left(-\frac{\tilde{\mu}^t(\theta;\omega)}{\tilde{\sigma}^t(\theta;\omega)}\right)}
            &= \exp\left(\frac{\tilde{v}^t(\omega)(2\tilde{\mu}^t(\theta;\omega)-\tilde{v}^t(\omega))}{2\tilde{\sigma}^t(\theta;\omega)^2}\right) \to \infty, \\
            \frac{\phi\left(\frac{\tilde{v}^t(\omega)-\tilde{\mu}^t(\theta;\omega)}{\tilde{\sigma}^t(\theta;\omega)}\right)}{\phi\left(-\frac{\tilde{v}^t(\omega)+\tilde{\mu}^t(\theta;\omega)}{\tilde{\sigma}^t(\theta;\omega)}\right)}
            &= \exp\left(\frac{2\tilde{v}^t(\omega)\tilde{\mu}^t(\theta;\omega)}{\tilde{\sigma}^t(\theta;\omega)^2}\right) \to \infty.
        \end{align*}
        Thus, among the PDF terms, only $\phi\left(\frac{\tilde{v}^t(\omega)-\tilde{\mu}^t(\theta;\omega)}{\tilde{\sigma}^t(\theta;\omega)}\right)$ is the dominant contribution, and $\widetilde{\mathrm{EI}}^t(\theta;\omega)$ reduces in the limit to
        \begin{align*}
            \widetilde{\mathrm{EI}}^t(\theta;\omega) 
            &\to 
            (v^t(\omega) - \mu^t(\theta;\omega)) \cdot \Phi\left( \frac{v^t(\omega) - \mu^t(\theta;\omega)}{\sigma^t(\theta;\omega)} \right) + \sigma^t(\theta;\omega) \cdot \phi\left( \frac{v^t(\omega) - \mu^t(\theta;\omega)}{\sigma^t(\theta;\omega)} \right) \\
            &= \mathrm{EI}^t(\theta;\omega).
        \end{align*}
        Then, under \Cref{siam:assump:rootless-globopt} and \Cref{siam:assump:rootless-unique},
        \[
            \left\|\arg\max_{\theta\in\Theta}\widetilde{\mathrm{EI}}^t(\theta;\omega) - \arg\max_{\theta\in\Theta}\mathrm{EI}^t(\theta;\omega)\right\| \to 0, \quad \text{as } t \to \infty.
        \]
        
        \paragraph{Case 2. Small-$\epsilon$}
        
        We now consider the small-$\epsilon(\omega)$ case. Let $\Theta^*_t \subseteq \Theta$ be a compact neighborhood containing the maximizer of the acquisition function. From \Cref{siam:assump:asymptotic-order}, we can neglect the contribution of $\mu^t(\theta;\omega)$, i.e., $v^t(\omega) \pm \mu^t(\theta;\omega) = v^t(\omega)$, where we can write
        \begin{align*}
            \widetilde{\mathrm{EI}}^t(\theta;\omega)
            &= \tilde{v}^t(\omega) \left[
                \Phi\left( \frac{\tilde{v}^t(\omega)}{\tilde{\sigma}^t(\theta;\omega)} \right)
              - \Phi\left( \frac{-\tilde{v}^t(\omega)}{\tilde{\sigma}^t(\theta;\omega)} \right)
            \right] 
            \\
            &+ \tilde{\mu}^t(\theta;\omega) \left[
                1
              - \Phi\left( \frac{\tilde{v}^t(\omega)}{\tilde{\sigma}^t(\theta;\omega)} \right)
              - \Phi\left( \frac{-\tilde{v}^t(\omega)}{\tilde{\sigma}^t(\theta;\omega)} \right)
            \right] \\
            &\quad - \tilde{\sigma}^t(\theta;\omega) \left[
                2\phi(0)
              - \phi\left( \frac{\tilde{v}^t(\omega)}{\tilde{\sigma}^t(\theta;\omega)} \right)
              - \phi\left( \frac{-\tilde{v}^t(\omega)}{\tilde{\sigma}^t(\theta;\omega)} \right)
            \right] \\
            &= 2\tilde{v}^t(\omega) \cdot \Phi\left( \frac{\tilde{v}^t(\omega)}{\tilde{\sigma}^t(\theta;\omega)} \right)
              + \tilde{\mu}^t(\theta;\omega) \cdot \left[ 1 - 2\Phi\left( \frac{\tilde{v}^t(\omega)}{\tilde{\sigma}^t(\theta;\omega)} \right) \right]
            \\
            &- 2\tilde{\sigma}^t(\theta;\omega) \cdot \left[ \phi(0) - \phi\left( \frac{\tilde{v}^t(\omega)}{\tilde{\sigma}^t(\theta;\omega)} \right) \right].
        \end{align*}
        Similarly, the standard EI can be expressed as
        \[
            \mathrm{EI}^t(\theta;\omega) = \tilde{v}^t(\omega) \cdot \Phi\left( \frac{\tilde{v}^t(\omega)}{\tilde{\sigma}^t(\theta;\omega)} \right) + \tilde{\sigma}^t(\theta;\omega) \cdot \phi\left( \frac{\tilde{v}^t(\omega)}{\tilde{\sigma}^t(\theta;\omega)} \right),
        \]
        so that
        \[
            \widetilde{\mathrm{EI}}^t(\theta;\omega) = 2 \cdot \mathrm{EI}^t(\theta;\omega) + \tilde{\mu}^t(\theta;\omega) \cdot \left[ 1 - 2\Phi\left( \frac{\tilde{v}^t(\omega)}{\tilde{\sigma}^t(\theta;\omega)} \right) \right] - 2\tilde{\sigma}^t(\theta;\omega) \cdot \phi(0).
        \]
        The contribution of $\tilde{\mu}^t(\theta;\omega) \cdot \left[ 1 - 2\Phi\left( \frac{\tilde{v}^t(\omega)}{\tilde{\sigma}^t(\theta;\omega)} \right) \right]$ is negligible under \Cref{siam:assump:asymptotic-order} as $t \to \infty$, since $|\tilde{v}^t(\omega)| \gg \tilde{\mu}^t(\theta;\omega)$. Under \Cref{siam:assump:local-variance}, the term $-2\tilde{\sigma}^t(\theta;\omega)\phi(0)$ is asymptotically constant over $\Theta^*_t$, so that $\widetilde{\mathrm{EI}}^t(\theta;\omega)$ is a positive linear transformation of $\mathrm{EI}^t(\theta;\omega)$ in the limit. Both then share the same maximizer, and under \Cref{siam:assump:rootless-globopt} and \Cref{siam:assump:rootless-unique},
            \[        \left\|\arg\max_{\theta\in\Theta}\widetilde{\mathrm{EI}}^t(\theta;\omega) - \arg\max_{\theta\in\Theta}\mathrm{EI}^t(\theta;\omega)\right\| \to 0, \quad \text{as } t \to \infty.
        \]
\end{proof}

\subsection{Numerical Experiment}

In this section, we have analyzed the asymptotic behavior of the root finding acquisition functions under the rootless regime. As a follow-up experiment, we empirically validate the established results through a simple one-dimensional example in which the signed discrepancy remains strictly positive over the parameter space. We define a noise-free signed discrepancy as 
\begin{equation*}
    \tilde{f}(\theta) := \theta^2 + \epsilon, \quad \theta \in [-1,1],
\end{equation*}
where $\epsilon > 0$ denotes the irreducible discrepancy. We consider two distinct regimes: a small-$\epsilon$ case with $\epsilon = 0.1$ and a large-$\epsilon$ case with $\epsilon = 10$. In both cases, the global optimizer is $\theta^* = 0$. For each $\theta \in \Theta$, noisy observations are generated as
\begin{equation*}
    \tilde{f}_{(j)}(\theta) = \tilde{f}(\theta) + \varepsilon_{(j)}, \quad \varepsilon_{(j)} \sim \mathcal{N}(0, 0.01^2),
\end{equation*}
where $\varepsilon_{(i)}$ denotes the $i$-th replication. At each $\theta$, we perform 10 independent replications. 

\begin{figure}[htbp]
    \centering
    \includegraphics[width=0.8\linewidth]{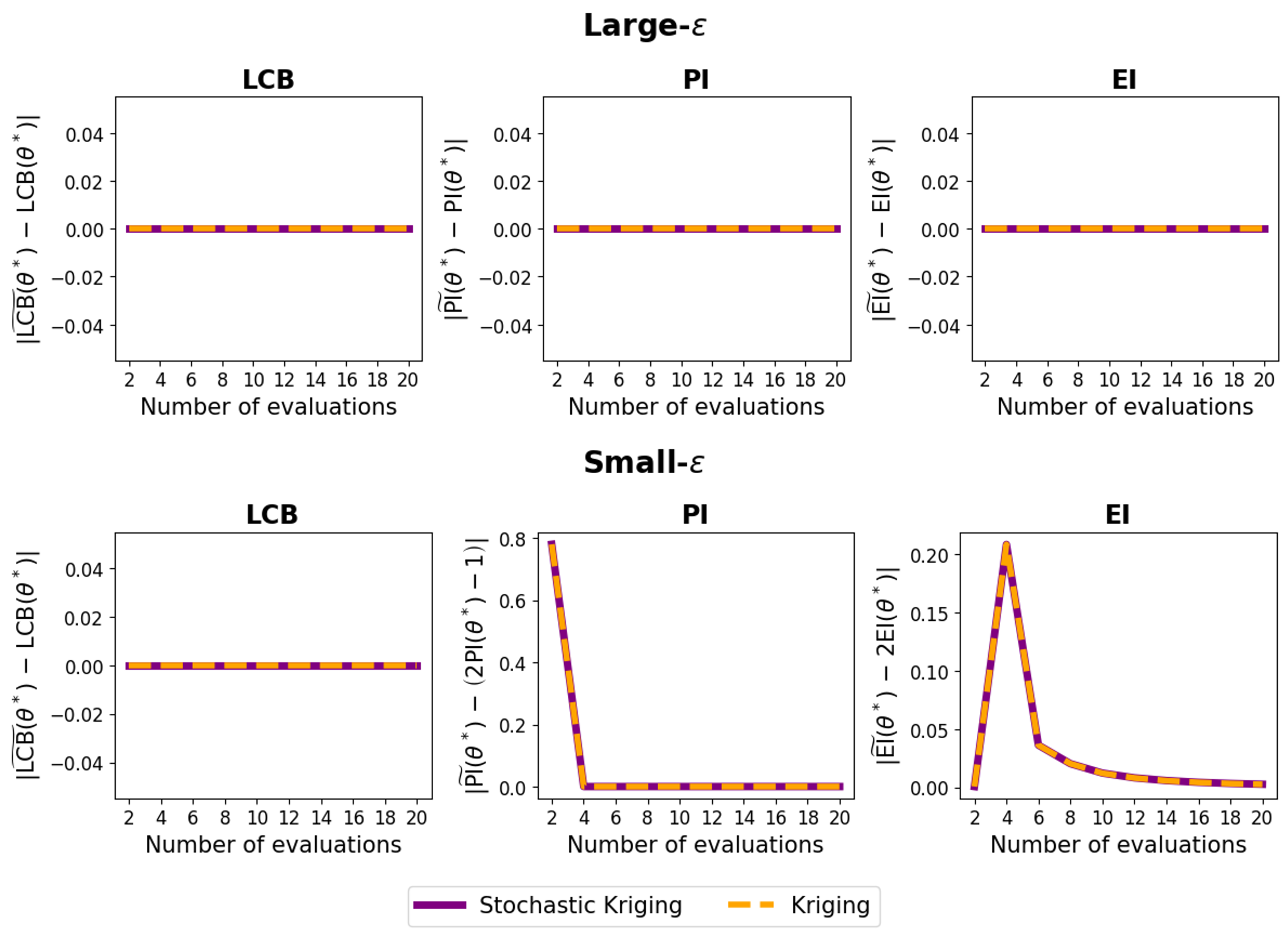}
        \caption{Acquisition function difference between the standard and root finding acquisition functions for large-$\epsilon$ and small-$\epsilon$ cases, averaged over 100 replications.}
    \label{siam:fig:rootless_result}
\end{figure}

In this experiment setting, we progressively increase the number of uniformly spaced design points, fit the metamodel at each stage $t$, and evaluate both the minimization and root finding acquisition functions at the known solution $\theta^* = 0$. Then, we compute the difference of each root finding acquisition function from the asymptotic form derived in \Cref{siam:sec:rootless}, and repeat this process across 100 macro replications.

\Cref{siam:fig:rootless_result} illustrates the averaged differences as the number of design points increases. In the large-$\epsilon$ case, the difference is negligible across all three acquisition functions regardless of the number of design points. This is because large $\epsilon$ makes the changes from the root finding modifications (e.g., additional CDF term in $\widetilde{\mathrm{PI}}$) effectively negligible. This result is somewhat natural since when the function is far away from zero, finding the root and finding the minimum become nearly equivalent problems. In the small-$\epsilon$ case, the LCB difference vanishes immediately, as the predictive mean becomes strictly positive with fewer observations in this simple example. For PI, the difference vanishes as more design points are added. For EI, however, the difference between the asymptotic relationship reduces more slowly, since it is further governed by the additional term $-2\tilde{\sigma}^t(\theta;\omega)\phi(0)$, which stabilizes only as the predictive uncertainty itself uniform within the region with many design points. This is also consistent with the discussion following \Cref{siam:assump:local-variance}, where the condition is expected to hold only at a relatively later stage of the search.



\section{Derivation of Stochastic Kriging Predictive Distribution}\label{siam:sec:SK-pred-derivation}
Consider a linear predictor of the form
\begin{equation*}
    \hat{\mu}(\theta) = w(\theta)^\top \bar{\boldsymbol{f}}^t_n,
\end{equation*}
where $w(\theta) \in \mathbb{R}^{p+t}$ is a weight vector. Then, the mean squared prediction error is given by
\begin{equation*}
    \mathbb{E}\left[(\bar{f}(\theta) - \hat{\mu}(\theta))^2\right]
    =
    k(\theta,\theta;l)
    -
    2w(\theta)^\top \mathbf{k}(\theta,\Theta^t;l)
    +
    w(\theta)^\top \left[\mathbf{k}(\Theta^t, \Theta^t;l) + \Sigma^{(t)}\right] w(\theta).
\end{equation*}
Minimizing this with respect to $w(\theta)$ yields the optimal weight
\begin{equation*}
    w^*(\theta) = \left[\mathbf{k}(\Theta^t, \Theta^t;l) + \Sigma^{(t)}\right]^{-1} \mathbf{k}(\theta,\Theta^t;l),
\end{equation*}
and the corresponding optimal predictor
\begin{equation*}
    \mu^t_n(\theta) = {w^*(\theta)}^\top \bar{\boldsymbol{f}}^t_n 
    = 
    \mathbf{k}(\theta, \Theta^t;l) \left[\mathbf{k}(\Theta^t,\Theta^t;l) + \Sigma^{(t)}\right]^{-1} \bar{\boldsymbol{f}}^t_n.
\end{equation*}
The minimum achievable prediction error is then
\begin{equation*}
    (\sigma^t_n)^2(\theta) = k(\theta, \theta;l) 
    - \mathbf{k}(\theta,\Theta^t;l) \left[\mathbf{k}(\Theta^t,\Theta^t;l) + \Sigma^{(t)}\right]^{-1} 
    \mathbf{k}(\theta,\Theta^t;l)^\top.
\end{equation*}
Under the GRF assumption, the predictive distribution at $\theta$ is given by
\begin{equation*}
    \eta^t_n(\theta) \sim \mathcal{N}\left(\mu^t_n(\theta), (\sigma^t_n)^2(\theta)\right).
\end{equation*}

\section{Derivation of Probability of Root Existence}\label{siam:sec:prob-rootexist}

Let $(\Omega, \mathcal{F}, \mathbb{P})$ be the probability space introduced in \Cref{siam:sec:rf-metamodeling}. The sign-change event for a pair $(\theta_a, \theta_b)$ is defined as
\begin{equation*}
    \mathcal{I}^{\pm}(\theta_a, \theta_b) 
    := \{\omega \in \Omega : \tilde{f}_n(\theta_a;\omega)\cdot\tilde{f}_n(\theta_b;\omega) < 0\},
\end{equation*}
under which, by \Cref{siam:thm:bolzano}, a root of $\tilde{f}_n(\theta)$ is guaranteed within $[\theta_a, \theta_b]$. Then, the probability of this event is
\begin{align*}
    \mathbb{P}(\mathcal{I}^{\pm}(\theta_a,\theta_b)) 
    &= \mathbb{P}(\tilde{f}_n(\theta_a)<0,\tilde{f}_n(\theta_b)>0) 
    + \mathbb{P}(\tilde{f}_n(\theta_a)>0,\tilde{f}_n(\theta_b)<0),
\end{align*}
since the sign-change event $\mathcal{I}^{\pm}(\theta_a, \theta_b)$ arises from the two opposite cases. As $\tilde{f}_{n(\theta)}(\theta)$ is a replication average, its exact distribution is generally intractable, and hence we approximate it using the stochastic kriging posterior $\tilde{\eta}^t_n(\theta) \sim \mathcal{N}(\tilde{\mu}^t_n(\theta), (\tilde{\sigma}^t_n)^2(\theta))$ and assume independent replications between configurations. Then, we can write the closed form expression as
\begin{align*}
    \mathbb{P}(\mathcal{I}^{\pm}(\theta_a,\theta_b)) 
    &\approx 
    \mathbb{P}(\tilde{\eta}^t_n(\theta_a)<0)\mathbb{P}(\tilde{\eta}^t_n(\theta_b)>0) 
    + \mathbb{P}(\tilde{\eta}^t_n(\theta_a)>0)\mathbb{P}(\tilde{\eta}^t_n(\theta_b)<0)
    \\
    &= 
    \Phi\left(-\frac{\tilde{\mu}^t_n(\theta_a)}{\tilde{\sigma}^t_n(\theta_a)}\right) 
    + \Phi\left(-\frac{\tilde{\mu}^t_n(\theta_b)}{\tilde{\sigma}^t_n(\theta_b)}\right) 
    - 2\Phi\left(-\frac{\tilde{\mu}^t_n(\theta_a)}{\tilde{\sigma}^t_n(\theta_a)}\right)
    \Phi\left(-\frac{\tilde{\mu}^t_n(\theta_b)}{\tilde{\sigma}^t_n(\theta_b)}\right).
\end{align*}

\section{Figures for Numerical Experiments}
This section presents the supplementary figures for \Cref{siam:sec:experiment}.

\begin{figure}[!htbp]
    \centering
    \includegraphics[width=0.40\linewidth]{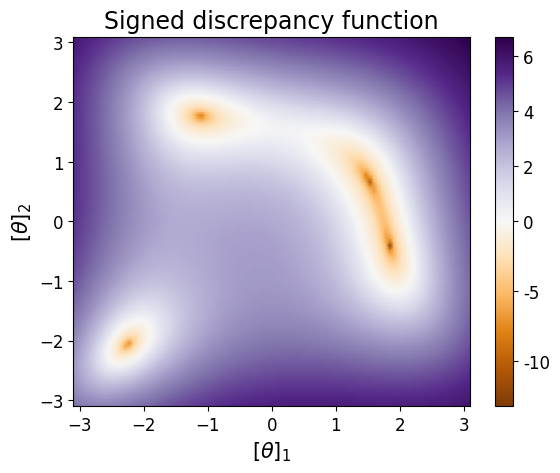}
    \caption{Noise-free signed discrepancy surface of the modified Himmelblau function.}
    \label{siam:fig:2d_func}
\end{figure}

\begin{figure}[htbp]
    \centering
    \includegraphics[width=0.70\linewidth]{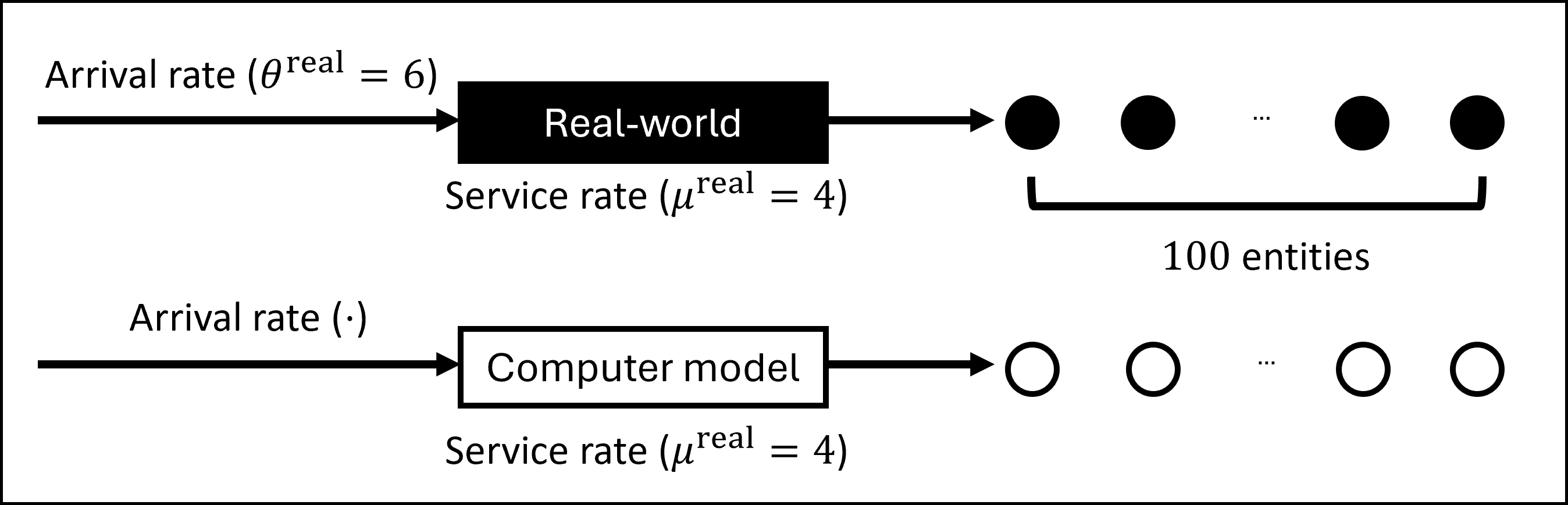}
    \caption{Overview of the M/M/1 queueing system calibration setup.}
    \label{siam:fig:mm1}
\end{figure}

\begin{figure}[htbp]
    \centering
    \includegraphics[width=0.77\linewidth]{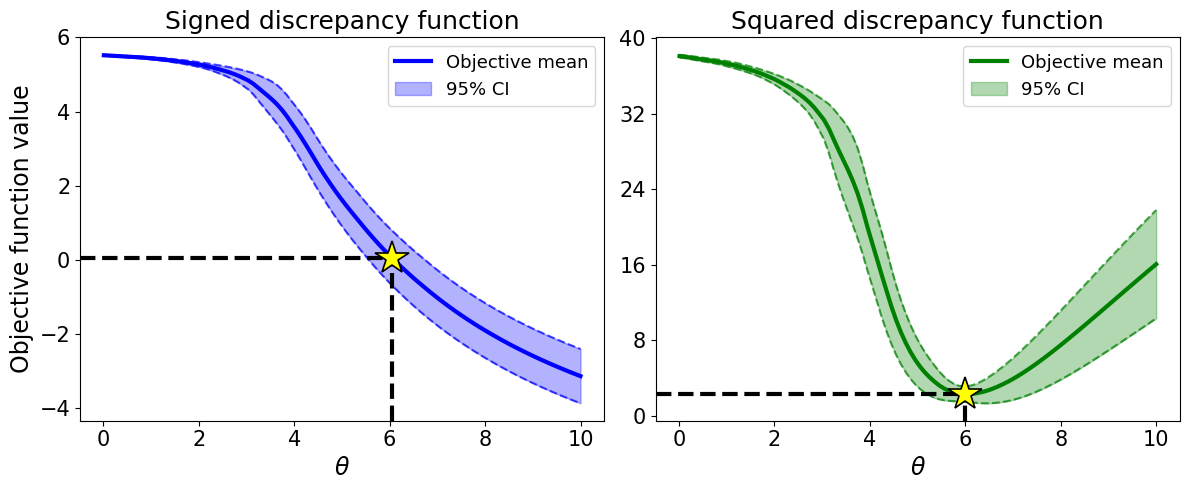}
    \caption{Estimated discrepancy functions for the M/M/1 calibration problem under root finding and minimization. For each $\theta$, the 95\% confidence interval is constructed from 100 macro replications.} 
    \label{siam:fig:mm1_func}
\end{figure}

\begin{figure}[htbp]
    \centering
    \includegraphics[width=0.45\linewidth]{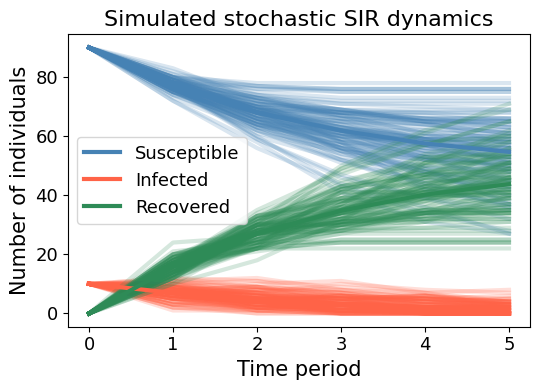}
    \caption{100 simulated realizations of the stochastic SIR model under the described experimental setting.}
    \label{siam:fig:epi_example}
\end{figure}

\begin{figure}[htbp]
    \centering
    \includegraphics[width=0.75\linewidth]{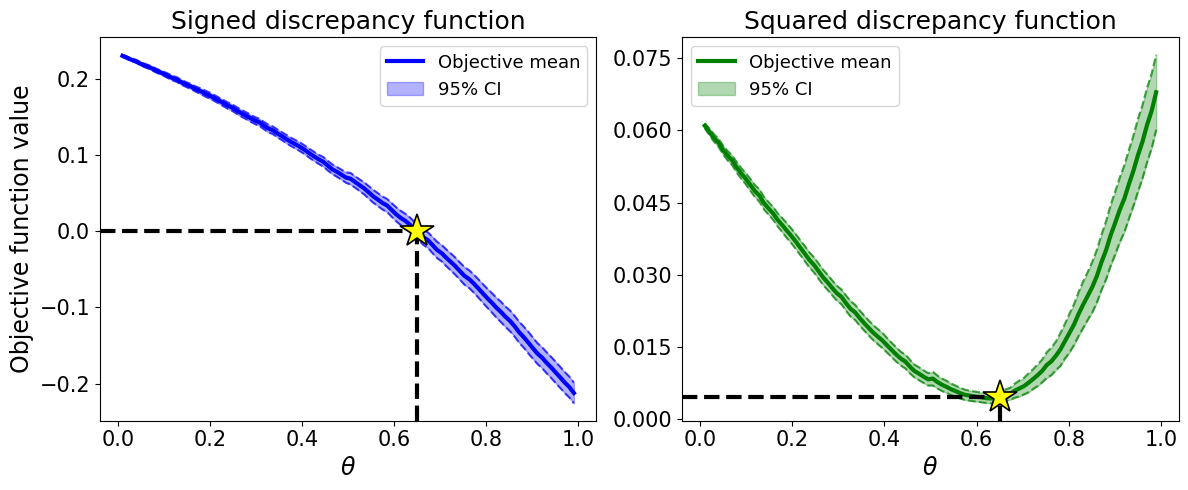}
        \caption{Estimated discrepancy functions for stochastic SIR calibration task under root finding and minimization. For each $\theta$, the 95\% confidence is constructed from 100 macro replications.}
    \label{siam:fig:epi_func}
\end{figure}


\bibliographystyle{unsrt}  
\bibliography{references}

\end{document}